\documentclass[nohyperref]{article}

\usepackage{microtype}
\usepackage{graphicx}
\usepackage{subfigure}
\usepackage{booktabs} 

\usepackage{hyperref}

\usepackage[accepted]{icml2023}

\usepackage{amsmath}
\usepackage{amssymb}
\usepackage{mathtools}
\usepackage{amsthm}

\usepackage[capitalize,noabbrev]{cleveref}

\usepackage{bm} 
\usepackage{placeins} 

\newcommand{\wt}[1]{\widetilde{#1}}
\newcommand{\Var}{\text{Var}}

\newcommand{\Prob}{\mathbb{P}}
\newcommand{\te}[1]{\text{\emph{#1}}}

\newcommand{\bbm}{\begin{bmatrix}}
\newcommand{\ebm}{\end{bmatrix}}
\newcommand{\ps}{\phantom{}}

\theoremstyle{plain}
\newtheorem{theorem}{Theorem}[section]

\newtheorem{lemma}[theorem]{Lemma}
\newtheorem{corollary}[theorem]{Corollary}
\theoremstyle{definition}
\newtheorem{definition}[theorem]{Definition}

\newtheorem{condition}[theorem]{Condition} 
\theoremstyle{remark}
\newtheorem{remark}[theorem]{Remark}

\usepackage[textsize=tiny]{todonotes}

\icmltitlerunning{Structure Learning of Latent Factors via Clique Search}

\begin{document}

\twocolumn[
\icmltitle{Structure Learning of Latent Factors via\\
           Clique Search on Correlation Thresholded Graphs}



\icmlsetsymbol{equal}{*}

\begin{icmlauthorlist}
\icmlauthor{Dale S. Kim}{ucla}
\icmlauthor{Qing Zhou}{ucla}
\end{icmlauthorlist}

\icmlaffiliation{ucla}{Department of Statistics and Data Science, University of California, Los Angeles, United States}

\icmlcorrespondingauthor{Qing Zhou}{zhou@stat.ucla.edu}

\icmlkeywords{Machine Learning, ICML, Factor Analysis, Structure Learning, Correlation Thresholding, Maximal Cliques}

\vskip 0.3in
]



\printAffiliationsAndNotice{}  

\begin{abstract}
Despite the widespread application of latent factor analysis, existing methods suffer from the following weaknesses: requiring the number of factors to be known, lack of theoretical guarantees for learning the model structure, and nonidentifiability of the parameters due to rotation invariance properties of the likelihood.
We address these concerns by proposing a fast correlation thresholding (CT) algorithm that simultaneously learns the number of latent factors and a rotationally identifiable model structure.
Our novel approach translates this structure learning problem into the search for so-called independent maximal cliques in a thresholded correlation graph that can be easily constructed from the observed data.
Our clique analysis technique scales well up to thousands of variables, while competing methods are not applicable in a reasonable amount of running time.
We establish a finite-sample error bound and high-dimensional consistency for the structure learning of our method.
Through a series of simulation studies and a real data example, we show that the CT algorithm is an accurate method for learning the structure of factor analysis models and is robust to violations of its assumptions.
\end{abstract}

\section{Introduction} \label{sec:introduction}

Factor analysis is a commonly used multivariate technique which conceptualizes observed variables as a function of unobserved latent factors.
Methods and discussions have appeared in a variety of fields, particularly the social sciences, such as psychology \citep{Reise2000}, sociology \citep{Werts1973}, education \citep{Schreiber2006}, and epidemiology \citep{Martinez1998}.
It is generally assumed that the number of latent factors is less than the number of observed variables, hence serving as a dimension reduction procedure in this sense.

To learn the parameters of factor analysis models, three problems must be addressed: (1) the number of latent factors must be determined, (2) the support of the coefficients must be found, and (3) a unique solution must be determined from rotationally equivalent parameters.
Prior work on learning factor analysis models typically use a constraint-based or a score-based approach.
Constraint-based methods involve analyzing permutations of correlations and partial correlations among the observed variables for constraints that would be implied by potential models \citep{Scheines1998, Silva2006}.
However, we note that the focus of these algorithms is to construct equivalence classes of possible models and can be computationally demanding.
In contrast, our goal is to develop efficient methods for learning and estimating a single model output in this work.

Score-based methods are generally more amenable to single model outputs.
Traditional Exploratory Factor Analysis (EFA) typically maximizes the likelihood, restricting the latent factors to be orthogonal.
An oblique factor solution can be extracted by rotating the orthogonal solution, subject to the model constraints.
There are numerous procedures for such rotations, which typically yield different solutions \citep[for a review of such methods see][]{Browne2001}.
After rotation, additional structure may be learned by setting small elements of $\Lambda$ to zero if they are below an ad-hoc threshold \citep{Ford1986, Howard2016}.
The major criticisms of EFA are the subjective use of these rotation criteria and thresholding steps, and requiring the number of latent factors to be known a priori.

As a potential solution to these problems in EFA, penalized methods also have been developed.
Most relevant to oblique factor analysis models are adding LASSO \citep{Tibshirani1996} and MCP \citep{Zhang2010} penalties to the likelihood, which were developed by \citet*{Hirose2014b}.
Instead of rotating factor coefficients after maximizing the likelihood, penalized EFA can achieve sparse solutions by directly maximizing a penalized likelihood.
This requires the use of tuning parameters, followed by model selection with the Bayesian Information Criterion (BIC) or cross-validation \citep[CV;][]{Scharf2019}.
However, a search over a large set of tuning parameters is often computationally intense.
Furthermore, the number of latent variables is still required as an input, and theoretical guarantees for rotational identifiability and structural estimation consistency have yet to be established.

For choosing the number of latent factors, many ad-hoc methods have been suggested, but suffer from poor performance, poor theoretical motivation, or both.
Most classical methods are related to the eigenvalues of the sample correlation matrix among the observed variables.
Famous examples include the Kaiser-Guttman criterion \citep{Guttman1954, Kaiser1960}, the Scree Test \citep{Cattell1966, Raiche2013}, and variants thereof \citep{Horn1965, Glorfeld1995}.
On the other hand, modern methods use a model selection approach \citep{Preacher2013}.
However, none of these methods are without controversy, and a great deal of literature has been devoted to criticisms on both empirical and theoretical grounds \citep{Browne1968,Ford1986, Zwick1986, Velicer1990, Howard2016, Auerswald2019}.

In summary, all methods of learning factor analysis must address three fundamental issues: (1) determine the number of factors, (2) learn the structure of the model, and (3) resolve the rotational nonidentifiability issue.
As we have reviewed, an overabundance of literature has been dedicated to addressing these issues \textit{separately}, all with varying degrees of success.
In contrast, we seek to address all three aforementioned issues \textit{simultaneously} under a unified framework.
We do this by making use of thresholded correlation graphs of the observed correlation matrix, and exploiting two common assumptions in factor analysis designs.
First, we assume that the correlation between variables that share latent factor parents is higher than the correlation between variables that do not.
Second, we assume that each latent variable has at least one observed variable of which it is the sole parent.
Under these conditions, there is a perfect correspondence between latent factors and a specific type of maximal clique from these graphs, which we call independent maximal clique (defined in Section~\ref{sec:ct_overview}).
Therefore, the structure learning problem is converted to a search for all independent maximal cliques in the graph.
We leverage this key relation to make the following contributions:
\begin{enumerate}
  \item We propose a computationally efficient algorithm for learning the number of latent factors and the support of the coefficients simultaneously.
  \item We establish high-dimensional consistency of our algorithm for learning the structure of the model.
  \item We demonstrate the efficacy and practical uses of our algorithm on both real and simulated data, including high-dimensional settings.
\end{enumerate}

There is another recent study that has taken a clique analysis approach to learning the structure of independent latent factors \citep{Markham2020}.
This work utilizes the maximal cliques of a conditional independence graph for structure learning.
In contrast, we allow for correlated latent factors and our method is much more computationally efficient through the use of independent maximal cliques.
Further, we establish theoretical guarantees for our method in high-dimensions.

Notation throughout this article will be as follows.
Define $[n] \coloneqq \{1, \ldots, n\}$.
Let $A \subseteq [n]$ and $B \subseteq [p]$ be index sets. 
The complement of $A$ is denoted as $A^c$.
For a matrix $M = (m_{ij}) \in \mathbb{R}^{n \times p}$, we define $M_{AB}$ as the submatrix of $M$ consisting of the rows indexed by $A$ and columns indexed by $B$.
Similarly for a vector $V \in \mathbb{R}^{n}$, we define $V_A$ as the subvector of $V$ consisting of the entries indexed by $A$.
We denote the support of $M$ as $\mathcal{A}(M) \coloneqq \{(i, j) : m_{ij} \neq 0\}$.
We use $\bm{0}$ to represent a matrix or vector of zeroes, whose dimension can be inferred from context and $I_n$ denotes the $n \times n$ identity matrix.

For graph theoretic notation, we define a graph $\mathcal{G}$ as an ordered pair $(V, E)$, explicitly denoted as $\mathcal{G}(V, E)$, where $V$ is a set of vertices and $E \subseteq V \times V$ is a set of edges.
For convenience, we will use $V = X$ to mean that the elements of the vertex set $V$ represent the index set of the random vector $X$.
We also restrict our attention to \textit{undirected} graphs.
A \textit{clique} of $\mathcal{G}(V, E)$ is a subset of vertices $C \subseteq V$ such that all pairs of distinct vertices in $C$ are connected by an edge.
Finally, a \textit{maximal clique} is a clique that cannot be extended by including more vertices from $V$.

\section{The Factor Analysis Model}

Let $X = (X_1, \dots, X_p) \in \mathbb{R}^{p}$ be a vector of observed variables.
The factor analysis model specifies the joint distribution of $X$ in the form of a structural equation model:
\begin{equation} \label{model}
X = \Lambda L + \epsilon,
\end{equation}
where $L = (L_1, \dots, L_d) \sim \mathcal{N}_d (0, \Phi)$ is a vector of latent variables or factors, $\epsilon = (\epsilon_1, \dots, \epsilon_p) \sim \mathcal{N}_p (0, \Omega)$ is a vector of independent errors with a diagonal $\Omega$, and $\Lambda = (\lambda_{ij}) \in \mathbb{R}^{p \times d}$ is a matrix of coefficients, or factor loadings.
For convenience, an additive mean vector $\mu$ is omitted from the model without loss of generality.
We assume that $d < p$, since factor analysis is generally used as a dimension simplification technique.
In the context of $\Lambda$, $X_i$ is a function of $L_j$ if and only if $\lambda_{ij} \neq 0$, in which case we may say that $L_j$ is a \textit{parent} of $X_i$ and $X_i$ a \textit{child} of $L_j$.
We assume that every $X_i$ has at least one parent, and every $L_j$ has at least one child, i.e., there are no rows or columns of full zeroes in $\Lambda$.
We are considering the more general case of oblique factor analysis models in this study, where the $L$ variables may be correlated.

The model stated in Equation~\eqref{model} implies a covariance structure $\Sigma$ for $X$ as follows:
\begin{equation} \label{eq:varexp}
\Sigma(\theta) \coloneqq \Var(X) = \Var(\Lambda L + \epsilon) = \Lambda \Phi \Lambda^T + \Omega,
\end{equation}
letting $\theta = \{\Lambda, \Phi, \Omega\}$.
We write $\Sigma(\theta)$ to make explicit that we are referring to $\Sigma$ as a function of the parameters $\Lambda$, $\Phi$, and $\Omega$.
At times, it will be easier to deal with observed variables which are unit variance scaled.
Let $D_{\sigma}=\text{diag}(\Sigma)^{1/2}$, i.e. a diagonal matrix with entries $\Sigma_{ii}^{1/2}$.
Then we define a unit variance scaled $X$ as $\wt{X}$ in the following manner:
\begin{equation}
\wt{X} \coloneqq D^{-1}_{\sigma} X = D^{-1}_{\sigma}(\Lambda L + \epsilon) = \wt{\Lambda} L + \wt{\epsilon},
\end{equation}
where $\wt{\Lambda}=D^{-1}_{\sigma}\Lambda$ and $\wt{\epsilon}=D^{-1}_{\sigma} \epsilon$.
Similarly, it follows that a correlation matrix $\wt{\Sigma}$ can be expressed as:
\begin{equation}\label{eq:Sigmaexp}
\wt{\Sigma}(\theta) \coloneqq D^{-1}_{\sigma} \Sigma D^{-1}_{\sigma} = \widetilde{\Lambda} \Phi \widetilde{\Lambda}^T + \widetilde{\Omega},
\end{equation}
where $\wt{\Omega} = D^{-1}_{\sigma} \Omega D^{-1}_{\sigma}$.
Note that the factor analysis model for $\Sigma$ and $\wt{\Sigma}$ are often used interchangeably, and the elements of $\wt{\Sigma}(\theta)$ may be referred to as $\rho_{ij}$.
Finally, notice that the structure of a factor analysis model is entailed by the number of factors $d$ and the support of $\Lambda$, denoted $\mathcal{A}(\Lambda)$.
Therefore we will define the \textit{structure} of a factor analysis model as the pair $(d, \mathcal{A}(\Lambda))$.

Given the structure of a factor analysis model $(d, \mathcal{A}(\Lambda))$, maximum likelihood is most widely used for estimating the parameters, based on
the Gaussian log-likelihood for $X~\sim~\mathcal{N}_p (0, \Sigma(\theta))$.
However, there is no closed-form solution for the MLE \citep{Joreskog1967}.
Therefore, iterative algorithms, such as Newton-Raphson \citep{Jennrich1969} or Expectation-Maximization \citep{Rubin1982}, are employed, which can be computationally intensive when the number of observed variables $p$ is large.
Furthermore, the parameters $\Lambda$ and $\Phi$ as in Equation~\eqref{eq:varexp} are in general not identifiable, often referred to as rotational nonidentifiability in the literature \citep{Anderson1956}.
This issue must be taken care of with additional criteria for parameter estimation or restrictions on the model structure.

\section{The Correlation Thresholding Algorithm} \label{sec:ct_algorithm}

\subsection{Preliminaries and Overview} \label{sec:ct_overview}

The main idea behind our algorithm is that for several broad classes of factor analysis models, the correlation between observed variables that share parents is stronger than correlations between variables that do not (these classes of models are discussed in Section~\ref{sec:assumptions}).
Subsequently, the correlation graph amongst the variables that share parents yields much information about the structure of the model.
We leverage these two ideas into an efficient algorithm to learn the structure.

Recall that $\rho_{ij}$ is the correlation between $X_i$ and $X_j$ given by $\widetilde{\Sigma}(\theta)$ in Equation~\eqref{eq:Sigmaexp}.
Our first step is to define a \textit{thresholded correlation graph} $\mathcal{G}(X, E(\tau))$ given some $\tau\in[0,1]$, where the edge set
\begin{equation} \label{eq:thresholded_graph}
E(\tau) \coloneqq \{(i, j) :  |\rho_{ij}| > \tau\}.
\end{equation}
In practice, given a sample of $X$, we can define an estimate of $E(\tau)$ as
\begin{equation} \label{eq:E_hat}
\hat{E}(\tau) \coloneqq \{(i, j) : \lvert r_{ij} \rvert > \tau\},
\end{equation}
where $r_{ij}$ denotes the sample correlation.

Given a thresholded correlation graph, an implied structure can be extracted by examining the cliques of the graph.
Specifically, there is a correspondence between the latent variable structures and a particular kind of maximal clique, which we term as \textit{independent maximal clique}:
\begin{definition}[Independent Maximal Clique]
\label{def:independent_maximal_clique}
Let $\mathcal{C} = \{C_1, \dots, C_k\}$ be the set of all maximal cliques in a graph $\mathcal{G}$.
Then, $C_i$ is an independent maximal clique if
\begin{equation}
C_i \nsubseteq \bigcup_{j \neq i} C_j.
\end{equation}
\end{definition}

Essentially, an independent maximal clique is a maximal clique that contains a vertex that is not a member of any other maximal clique.
We call such a vertex a \textit{unique member} of the independent maximal clique.
We use the word ``independent'' as an analog to the notion of linear independence in a vector space.
That is, an independent maximal clique cannot be covered by the union of any of the other maximal cliques.
In Section~\ref{sec:structural_identifiability}, we show that the each independent maximal clique corresponds to a latent variable, whose children are the members of those cliques.
This transforms the result of the clique search into a factor analysis structure.

\subsection{The Algorithm} \label{sec:ct_algorithm_specific}

Putting these ideas together, the core task of our algorithm is to search for a suitable $\tau_0$.
This can be done by searching over a set of candidate thresholds $\tau_k \in [0, 1]$ and analyzing their respective thresholded correlation graphs $\mathcal{G}(X, \hat{E}(\tau_k))$ for independent maximal cliques.
We exploit the correspondence between these cliques and the factor analysis structure to learn the number of latent variables and the support of $\Lambda$.
When this is done over each candidate threshold, this yields a set of candidate models for which we can utilize model selection procedures (e.g., BIC) to select a final model.
We formally describe these steps in Algorithm~\ref{alg:ct}.

\begin{algorithm}[tb]
   \caption{The Correlation Thresholding Algorithm}
   \label{alg:ct}
\begin{algorithmic}[1]
   \STATE {\bfseries Input:} Sample correlation matrix $R$ and set of thresholds  $\tau = \{\tau_k : k \in [m]\}$.
   \STATE {\bfseries Output:} Parameter estimates $\hat{\theta}$.
   \FOR{$k \in [m]$}
   \STATE Calculate $\mathcal{G}(X, \hat{E}(\tau_k))$ and extract the set of independent maximal cliques: $\mathcal{C}_k = \{C_1, \dots C_{\lvert \mathcal{C}_k \rvert}\}$;
   \STATE Set $\hat{d}_k = \lvert \mathcal{C}_k \rvert$;
   \STATE Initialize $\hat{A}_k = \emptyset$;
   \FOR{$(i, j) \in [p] \times [\hat{d}_k]$}
   \STATE If $i \in C_j$, add $(i, j)$ to $\hat{A}_k$;
   \ENDFOR
   \STATE Estimate $\hat{\theta}_k$ given $(\hat{d}_k, \hat{A}_k)$, i.e., subject to $\lambda_{ij} = 0$ for all $(i, j) \notin \hat{A}_k$;
   \ENDFOR
   \STATE Select one of the $m$ estimates from $\{\hat{\theta}_k : k \in [m]\}$ via a model selection procedure.
\end{algorithmic}
\end{algorithm}

To quickly find all independent maximal cliques in a graph, we can employ the following Lemma.
\begin{lemma} \label{lem:indep_maximal_clique}
Given a graph $\mathcal{G}(X, E)$, let $\te{ne}(X_i)$ be the set of vertices that contains $X_i$ and every node that shares an edge with $X_i$ (the neighbors of $X_i$).
\begin{enumerate}
  \item If $\te{ne}(X_i)$ is a clique, then $\te{ne}(X_i)$ is also an independent maximal clique and $X_i$ is a unique member of this clique.
  \item If $C$ is an independent maximal clique, then $C=\te{ne}(X_i)$ for any unique member $X_i\in C$.
\end{enumerate}
\end{lemma}

In the worst case scenario, all independent maximal cliques can be found by checking whether $\te{ne}(X_i)$ is a clique for every node $X_i$.
The computational cost for checking if $\te{ne}(X_i)$ is a clique has a brute force complexity of $O(k^2)$, assuming a maximum neighbor size of $k$.
Thus, the total computational cost on all $p$ nodes can be no greater than and usually well below $O(k^2p)$, which is very efficient even for large graphs, allowing our algorithm to be used in high-dimensional settings.
This is in sharp contrast to the exponential complexity in listing all (non-independent) maximal cliques in a graph \citep{Eppstein2010}.

After extracting these independent maximal cliques in Step~4, we learn the structure of the model in Steps~5 through~9.
The number of independent maximal cliques is set as the estimate of $d$, which is also the number of columns in $\Lambda$.
Then, the nodes in each $C_j$ determine if $\lambda_{ij}$ is zero or non-zero for each $i\in[p]$, allowing us to construct a candidate support $\hat{A}_k$.

In Step~10 we estimate each model given the learned structure, then in Step~12 we use a model selection procedure to select one of the models.
We note that these steps are general in that any estimation and model selection method can be utilized here.
In our implementation, we will prefer to use maximum likelihood estimation and BIC for model selection.
Since this pair of methods are statistically consistent, this leads to the final output model having consistent parameters and model structure, as we will show in Section~\ref{sec:consistency}.

\section{Theoretical Guarantees} \label{sec:theoretical_analysis}

In this section, we establish theoretical guarantees for the CT algorithm.
We assume throughout that the factor analysis model in Equation~\eqref{model} holds. 
Proofs of these results can be found in Appendix~\ref{app:proofs}.

\subsection{Assumptions} \label{sec:structural_identifiability}

We first present the main assumptions under which the structure for $\Lambda$ can be recovered from the thresholded correlation graph.
A discussion of these assumptions is provided in Section~\ref{sec:assumptions}.

Let the \textit{parent set} of $X_i$ be $\Pi_i \coloneqq \{j : \lambda_{ij} \neq 0, j \in [d]\}$.
Then we formalize the set of pairs that share parents as
\begin{equation} \label{e_definition_E}
E_0 \coloneqq \{ (i, j) \in [p] \times [p] : \Pi_i \cap \Pi_j \neq \emptyset\}.
\end{equation}
Subsequently, we denote the set of pairs that do not share parents (the complement of $E_0$) as
\begin{equation} \label{e_definition_Ec}
E_0^c = \{ (i, j)\in [p] \times [p] : \Pi_i \cap \Pi_j = \emptyset\}.
\end{equation}
Essentially, we would like to find some threshold $\tau_0$ that is able to separate the $E_0$ and $E_0^c$ sets by the magnitude of the correlations.
We will define this notion as \textit{thresholdable}:
\begin{definition}[Thresholdable]
\label{def:thresholdable}
A set of parameters $\theta$ is called \textit{thresholdable} if there exists a threshold $\tau_0$ such that
\begin{equation} \label{e_thresholdability}
\underset{(k, l) \in E_0^c}{\max} \lvert \rho_{kl} \rvert < \tau_0 < \underset{(i, j) \in E_0}{\min}\lvert \rho_{ij} \rvert.
\end{equation}
\end{definition}

Recall the use of independent maximal cliques (Definition~\ref{def:independent_maximal_clique}) in the CT algorithm.
Perfect model structure recovery can be achieved if there is a one-to-one correspondence between the latent variable structures and the independent maximal cliques.
A simple sufficient condition for such a correspondence to hold is the \emph{unique child condition}:
\begin{condition}[Unique Child Condition] \label{con:ucc}
Let the \textit{child set} of a latent variable be denoted $\te{ch}(L_k) = \{i \in[p]: \lambda_{ik} \neq 0\}$.
If
\begin{equation} \label{eq:uc_def}
U_k \coloneqq \te{ch}(L_k) - \bigcup_{j \neq k} \te{ch}(L_j) \neq \emptyset,\quad\quad \forall\; k\in[d],
\end{equation}
i.e., if each latent variable $L_k$ has a non-empty set of unique children $U_k$,
then we say that the unique child condition holds.
It essentially means that all latent parents have at least one unique child variable.
\end{condition}
Given this condition, we can obtain a bijection between the latent variables and the independent maximal cliques in $\mathcal{G}(X, E_0)$.
We state this in the following lemma.
\begin{lemma} \label{lem:structural_identifiability}
If the unique child condition holds in $\Lambda$ (Condition~\ref{con:ucc}), then the set $\{\te{ch}(L_k):k\in[d]\}$ is identical to the set of independent maximal cliques in $\mathcal{G}(X, E_0)$.
\end{lemma}

Recall the key observation that the dimension of $L$ and support of $\Lambda$, i.e. $(d,\mathcal{A}(\Lambda))$, completely encodes the structure of the model.
The CT algorithm leverages Lemma~\ref{lem:structural_identifiability} to recover the structure of a factor analysis model $(d,\mathcal{A}(\Lambda))$ by finding independent maximal cliques in an estimated graph $\mathcal{G}(X, \hat{E}(\tau_k))$.

\begin{remark} \label{rem:permutation_equivalence}
We note our use of $\mathcal{A}(\Lambda)$ defines a model structure up to a column permutation of $\Lambda$.
That is, we consider different ordering or labeling of the factors to be equivalent, since they define the same $\Sigma(\theta)$ in Equation~\eqref{eq:varexp}.
\end{remark}

\subsection{Error Bounds and Consistency} \label{sec:consistency}

In this section, we establish the consistency of the CT algorithm.
We will call a structural estimate $(\hat{d}, \mathcal{A}(\hat{\Lambda}))$ consistent if
\begin{equation} \label{eq:struc_consistent}
\lim_{n \to \infty} \Prob\Big[ (\hat{d}, \mathcal{A}(\hat{\Lambda})) = (d, \mathcal{A}(\Lambda))\Big] = 1,
\end{equation}
given an i.i.d. sample of size $n$ from the model in Equation~\eqref{model}.
By Lemma~\ref{lem:structural_identifiability}, the model structure $(\mathcal{A}(\Lambda), d)$ can be recovered exactly from the set of independent maximal cliques in $\mathcal{G}(X,E_0)$ when the unique child condition holds.
Therefore, structural consistency holds when $\lim_{n \to \infty} \Prob(\hat{E}(\tau_0) = E_0) = 1$ under the unique child condition for a suitable $\tau_0$.
In what follows, it will be useful to define a gap of separation for a thresholdable $\theta$ as
\begin{equation} \label{eq:gap}
\gamma \coloneqq \dfrac{1}{2} \left[ \underset{(i,j) \in E_0}{\min}\lvert \rho_{ij} \rvert  - \underset{(i,j) \in E_0^c}{\max} \lvert \rho_{ij} \rvert \right].
\end{equation}

\begin{theorem} \label{theo:consistency}
\label{t_bound}
Assume the model described in Equation~\eqref{model} holds for $X$ and that the correlations between all pairs $(X_i, X_j)$ are bounded such that $\max_{i \neq j}\lvert \rho_{ij} \rvert \leq M < 1$.
If $\theta$ is thresholdable with a gap $\gamma>0$, then
\begin{equation} \label{eq:consistency_bound}
\Prob(\hat{E}(\tau_0) \ne E_0) \leq  C p(p - 1) (n - 2)\left( \dfrac{4 - \gamma^2}{4 + \gamma^2} \right)^{n - 4}:= \eta,
\end{equation}
where $0 < C < \infty$ only depends on $M$.
If additionally the unique child condition holds (Condition~\ref{con:ucc}), then we have
\begin{equation} \label{eq:support_bound}
\Prob((\hat{d}, \mathcal{A}(\hat{\Lambda})) = (d, \mathcal{A}(\Lambda))) \geq 1 - \eta,
\end{equation}
where $(\hat{d}, \mathcal{A}(\hat{\Lambda}))$ is the estimated model structure by the CT algorithm with cutoff $\tau_0$.
\end{theorem}

Due to the exponential decay of the term $[(4 - \gamma^2)/(4 + \gamma^2)]^{n-4}$, consistency is trivially implied under a fixed $p$ regime.
More generally speaking, for any joint distribution of $X$ under which the central limit theorem holds for the sample correlations $\{r_{ij}\}$, structural consistency would also follow.
By the classical central limit theorem and the delta method, this would include the class of distributions with finite fourth-order moments \citep{Ferguson1996}.
Furthermore, we will use the bound described in Inequality~\ref{eq:consistency_bound} to develop a consistency result with high-dimensional accommodations where the dimension $p = p_n\gg n$.
\begin{theorem} \label{theo:hd_consistency}
Assume the model described in Equation~\eqref{model} holds for $X$ and that the correlations between all pairs $(X_i, X_j)$ are bounded such that $\max_{i \neq j}\lvert \rho_{ij} \rvert \leq M < 1$ for some universal constant $M$ independent of $n$.
If $\theta$ is thresholdable with a gap $\gamma=\gamma_n$ such that $\gamma_n^2 \geq c_1/(n - 4)^b$ for some $c_1 > 0$ and $b \in [0, 1)$ when $n$ is large, and $p_n = o(\exp(c(n - 4)^{1-b}))$, where $0 < c < c_1/8$, then
\begin{equation}
\lim_{n \to \infty} \Prob(\hat{E}(\tau_0) = E_0) = 1.
\end{equation}
If additionally the unique child condition holds (Condition~\ref{con:ucc}), then the structural estimate is consistent, as in Equation~\eqref{eq:struc_consistent}.
\end{theorem}

Note that any fixed value between $\max_{E_0^c} \{ \lvert \rho_{ij} \rvert \}$ and $\min_{E_0}\{\lvert \rho_{ij}\rvert\}$ will be a valid choice for $\tau_0$ for structure learning consistency.
This result is straightforward to generalize to non-Gaussian forms of $X$, which could result from non-Gaussian combinations of $L$ and $\epsilon$.
All that would be required is to replace our use of Lemma~\ref{l_KB07} (a Gaussian sample correlation concentration bound) in the proofs of Theorems~\ref{theo:consistency} and~\ref{theo:hd_consistency} with a bound for any non-Gaussian $X$ of interest (see Appendices~\ref{sec:proof_consistency} and~\ref{sec:proof_hds}, respectively).
So long as this bound is sufficiently well-behaved, the probability bounds in Theorem~\ref{theo:consistency} will hold as will Theorem~\ref{theo:hd_consistency} with different dependencies between $p$ and $n$.

In the practical context of the CT algorithm, recall that a suitable $\tau_0$ is actually unknown, and the algorithm estimates and selects among a set of models based on a candidate set $\{\tau_k\}$.
Assuming that a suitable $\tau_0$ is contained in $\{\tau_k\}$, the unique child condition and consistency implies that the correct model structure is among the set of candidate models, asymptotically.
From here, overall parameter consistency follows by simply using a consistent parameter estimation method (Step~10) and a consistent model selection procedure (Step~12) in the algorithm.
A straightforward choice would be to use maximum likelihood estimation in conjunction with BIC model selection.
Then, asymptotically, the CT algorithm will produce the correct model structure with consistent parameter estimates.

\subsection{Rotational Uniqueness} \label{s_rotationalUniqueness}

An important consideration with a factor analysis model is the identifiability of the parameters $\theta = \{\Lambda, \Phi, \Omega\}$.
It is well known that factor analysis models lack of rotational uniqueness, which implies that there may be many $(\Lambda, \Phi)$ pairs that exist such that $\Sigma(\theta) = \Lambda \Phi \Lambda^T + \Omega$.
However, the solutions learned by the CT algorithm resolves this nonidentifiability issue given that the zero constraints implied by $\mathcal{A}(\hat{\Lambda})$ are preserved.
A formal definition of rotational uniqueness can be found in Appendix~\ref{sec:def_rot_unique}.

\begin{corollary} \label{cor:ucc_rotation}
If the unique child condition holds in $\Lambda$, then $\theta$ is locally rotationally unique (i.e., unique up to a polarity reversal on columns).
\end{corollary}

\begin{corollary} \label{cor:ct_local_rotation_unique}
Any $\hat{\theta}_k$ for $k \in [m]$, produced by Step~10 of the CT algorithm, is locally rotationally unique.
\end{corollary}

First we note that all matrix factorizations will have a polarity reversal on columns or rows as a source of non-uniqueness unless the signs of the main diagonal (or a permutation thereof) are fixed and non-zero.
Since the model in Equation~\eqref{model} makes no assumptions regarding the signs in $\Lambda$, local rotational uniqueness is the best type of rotational uniqueness that can be established.
Second, note that Corollary~\ref{cor:ct_local_rotation_unique} holds regardless if Condition~\ref{con:ucc} is true in the population structure.
Thus the CT algorithm can be used as a model approximation tool for finding locally rotationally unique structures.

\subsection{Discussion of Assumptions} \label{sec:assumptions}

We discuss the practicality of our thresholdability and unique child assumptions and how they relate to common factor analytic designs.
Regarding the thresholdability of $\theta$, several widely used factor analysis designs either meet the assumption outright, or under mild conditions.
These stem from a technical necessary and sufficient condition for thresholdability presented in Lemma~\ref{generalSufficient} in the Appendix.
Relevant to our discussion are the following corollaries to the lemma, which we discuss here.

\begin{corollary} \label{cor:orthogonal_factors}
If $\Phi = I_d$, then $\theta$ is thresholdable.
\end{corollary}

That is, if we have the orthogonal factor analysis design, then thresholdability is met.
Another common scenario is when $\Lambda$ has exactly one non-zero entry per row.
This is called ``independent cluster structure'' \citep{Harris1964} or ``perfect simple structure'' \citep{Jennrich2006}.
Such structures lead to a simplification of the thresholdability condition:

\begin{corollary} \label{cor:independent_cluster}
If $\Lambda$ has exactly one non-zero entry per row, then $\theta$ is thresholdable if
\begin{equation}
\underset{(k, l) \in E_0^c}{\max} \lvert \wt{\lambda}_{ke} \wt{\lambda}_{lf} \phi_{ef}  \rvert < \underset{(i, j) \in E_0}{\min} \lvert \wt{\lambda}_{ic} \wt{\lambda}_{jc} \rvert,
\end{equation}
where $\Pi_i = \Pi_j = \{c\}$, $\Pi_k = \{e\}$, and $\Pi_l = \{f\}$.
\end{corollary}

Corollaries~\ref{cor:orthogonal_factors} and~\ref{cor:independent_cluster} involve desirable properties of factor analytic designs.
It has been suggested that latent variable models should be designed such that the latent factors be distinguishable from one another, or that they are not too highly correlated \citep{Whitely1983}.
If the latent factors are too highly correlated, then a factor solution with less dimensions may be better suited.

As a common design in educational and psychological test construction \citep{Hattie1985, Anderson1988}, an independent cluster structure yields mutually exclusive subsets of children for each latent variable.
In other words, each observed variable provides a ``measurement'' of a single latent variable alone.
In contrast, our unique child condition (Condition~\ref{con:ucc}) is much more general, only requiring a single observed variable to serve as the sole measurement.
Many other latent variable algorithms only focus on the independent cluster structure \citep{Scheines1998, Jennrich2001, Jennrich2006, Silva2006}, or require 3 to 4 observed variables to serve as unique children \citep{Shimizu2009, Kummerfeld2016}.

Additionally, we examine the unique child condition under a random graph model for $\mathcal{A}(\Lambda)$, in which edges are independently connected between any $X_i$ and $L_j$ with probability~$\alpha$.
Let $m=d \alpha$ be the expected number of parents for any $X_i$. One can show that the unique child condition holds with probability $\geq 1-d \exp[-\alpha p(1-\alpha)^d]$.
Consequently, if $\log(\log d)\ll m\ll \min\{\log(m p/d),d\}$, then the unique child condition holds with high probability.

The assumption of the unique child condition does have a few limitations as it precludes certain structures from being perfectly discovered.
Examples of these structures are illustrated in Appendix~\ref{supp:figures} (Figure~\ref{fig:ucc_counterexamples}).
However, even under such settings, we will show through simulation that the CT algorithm will select a structure close to the true structure despite the unique child condition not holding.

\section{Simulation Studies} \label{sec:simulations}

\subsection{Low-Dimensional Settings}

Our first set of simulations were conducted in low-dimensional settings.
Here, we compared the CT algorithm against three other methods: (1) EFA, (2) EFA-LASSO, and (3) EFA-MCP (all described in Section~\ref{sec:introduction}).
Note that these EFA methods all require $d$ as an input, thus we use the CT algorithm to give these EFA methods a set of $d$ to work with.
This was to make the comparison as fair as possible, rather than than using ad hoc choices.
More specifically, we ran the CT algorithm to Step~5, where $d$ is estimated from the number of independent maximal cliques.
Thereafter, we replaced the support learning portion (Steps~6 through~9) of the algorithm with one of the EFA procedures.
Then the support of the model was saved from the EFA methods and resumed the algorithm from Step~12, where the MLE was estimated from the support and used for model selection.

We generated data sets from a zero-mean Gaussian distribution, with a covariance matrix $\Sigma$ parameterized by $\theta$.
The structure of $\Lambda$ followed an independent cluster structure (one non-zero entry per row).
We focused on this structure since it is the most common factor analysis design and it was the simulation design used in the studies proposing the penalized EFA methods \citep{Hirose2014a, Hirose2014b}.
The number of latent variables ($d$) was set to 3 with the number of children per latent variable set to 5 and $n = 1000$.
The non-zero entries of $\Lambda$ were drawn from a uniform distribution, $\lambda_{ij} \sim \text{Uniform}(0.6, 0.8)$.
We varied the magnitude of the off-diagonals in $\Phi$, as it is a key factor in whether or not $\theta$ is thresholdable for these structures, as shown by Corollary~\ref{cor:independent_cluster}.
Their entries began with the range of $[0.6, 0.8]$, with a low-magnitude setting scaling these by 0.25 and a high-magnitude setting scaling these by 0.75.
The tuning parameters of the penalized EFA methods were left at the software package defaults, which were 30 tuning parameters for EFA-LASSO and a set of 270 tuning parameters for EFA-MCP.
We conducted 100 replications per condition.
Further details of the software and data generating settings can be found in Appendix~\ref{supp:simulation}.

We examined three outcomes to assess the performance of the methods: (1) The $F_1$ score of $\mathcal{A}(\hat{\Lambda})$, (2) the learned number of latent variables, and (3) the computational efficiency of each method.
A precise definition of the $F_1$ score of $\mathcal{A}(\hat{\Lambda})$ can be found in Appendix~\ref{sec:outcomes}.
To measure computational efficiency we simply counted the number of models each method estimated.
This was to avoid idiosyncratic differences between the software implementations of each method.
For the CT algorithm, this is simply the number of unique structures obtained by the sequence of $\tau_k$.
For EFA, this translates to the number of unique $d$ obtained by the sequence of $\tau_k$.
For EFA-LASSO and EFA-MCP, this is the number of tuning parameter combinations to search over (30 for LASSO, 270 for MCP), per unique $d$ in the sequence of $\tau_k$.

The results of this simulation are displayed in Figure~\ref{fig:icml1}.
CT and EFA-MCP have the best $F_1$ scores (very close to 1.0), with EFA-LASSO at around 0.75 and EFA at 0.5 across both conditions.
All methods learned the number of latent variables correctly in all data sets, and thus were omitted from the figure.
For computational efficiency, the CT algorithm estimated a substantially less amount of models compared to the penalized EFA methods.
EFA showed the best computational efficiency, but in contrast had the worst $F_1$ score.
These results demonstrate that in low-dimensional settings, that the CT algorithm performs with near perfect accuracy along with EFA-MCP, however with substantial computational savings.

\begin{figure}[t]
\centering
\includegraphics[width=0.4\textwidth, keepaspectratio]{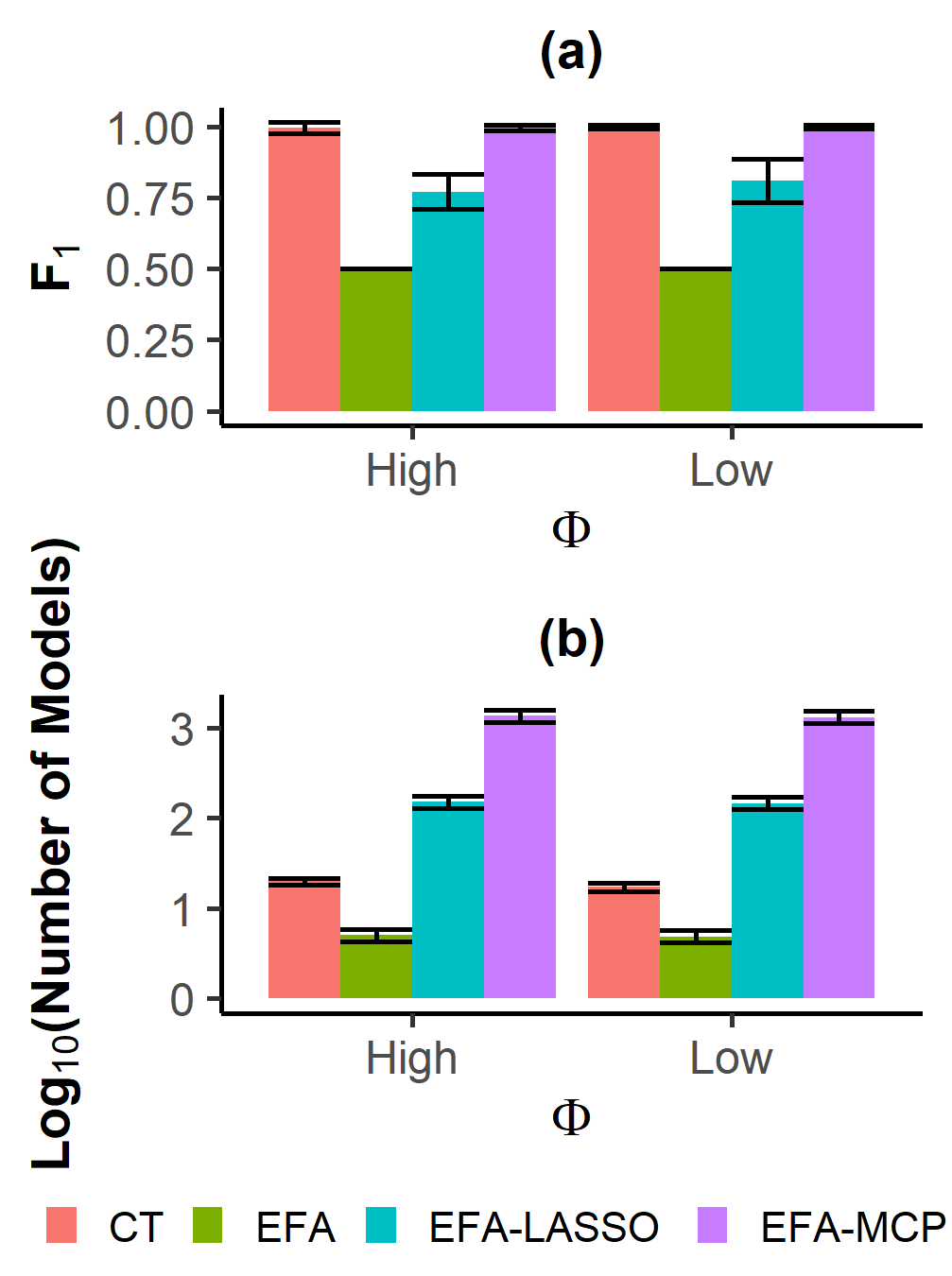}\\
\caption{Averages of the performance metrics for the low-dimensional simulation. Error bars depict $\pm1$ standard deviation.}
\label{fig:icml1}
\end{figure}

\subsection{High-Dimensional Settings} \label{sec:simulation_high_dim}

For the high-dimensional settings, we examined the scenario where both $p$ and $d$ grow proportionally with $n$, and $n < p$.
We examined three $(n, p, d)$ settings where $n \in\{ 250, 500, 1000\}$, and set $p = 1.5n$ and $d = 0.1n$.
In addition, we studied two conditions where the key assumptions of the CT algorithm would be violated: (1) thresholdability, which we violated using high-magnitude off-diagonals in $\Phi$ as in the previous simulation, and (2) the unique child condition.
We violated the unique child condition by starting with the independent cluster structure, then randomly selecting 75\% of the latent variables to have no unique children, whose children were all given an extra random parent.
To isolate the effect of the unique child condition from that of thresholdability, we ensured thresholdability was always met in the latter set of simulations by setting $\Phi = I_d$ (Corollary~\ref{cor:orthogonal_factors}).
Further details regarding the simulation settings can be found in Appendix~\ref{supp:simulation} and additional results using more varied assumption violations can be found in Appendix~\ref{supp:figures} (Figure~\ref{fig:thresh_ucc_range}).

Under these high $p$ settings, both the EFA and penalized EFA methods are prohibitively slow, thus could not be used as comparisons for this study.
Further, MLE routines also do not complete in a reasonable amount of running time, hence, we omitted the estimation step of the CT algorithm (Step~10).
Rather, a final model structure was chosen as the one with the minimum Hamming distance (HD) among the candidate thresholds $\{\tau_k$, $k\in[m]\}$ (a precise definition of HD can be found in Appendix~\ref{sec:outcomes}).
As before, we examined the $F_1$ score, $\hat{d}$, and computational efficiency as the outcomes for this study.

\begin{figure*}[t]
\centering
\includegraphics[width=0.9\textwidth, keepaspectratio]{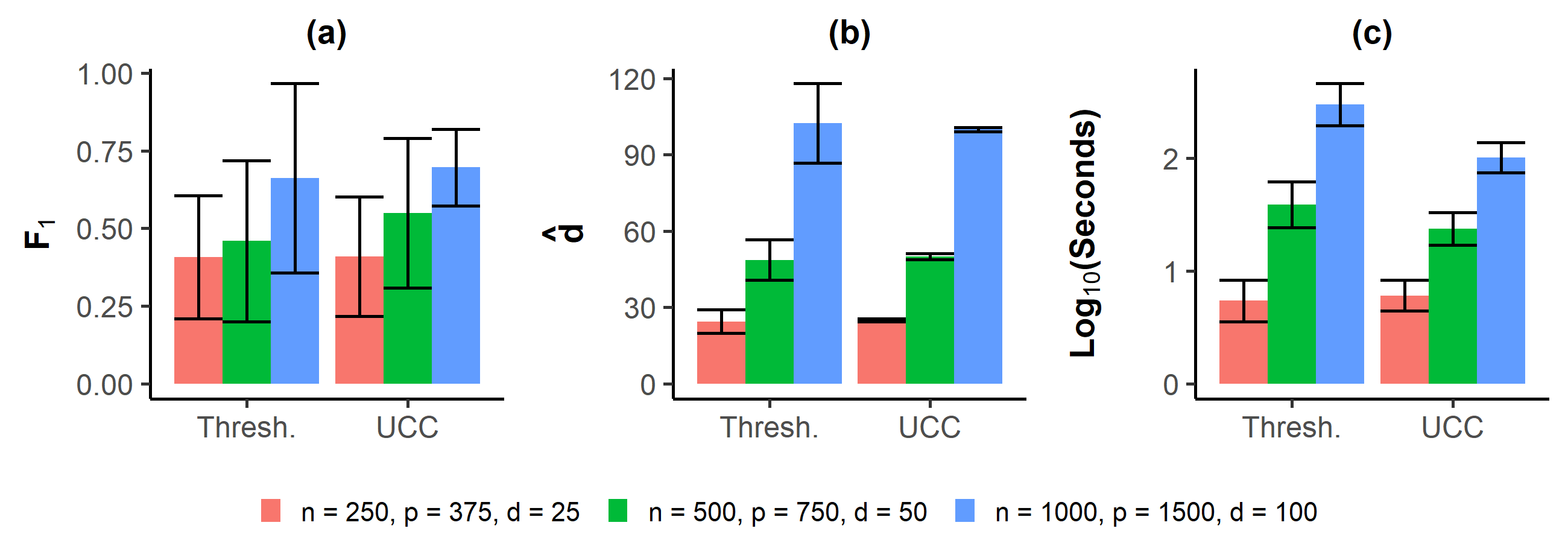}
\caption{Averages of the performance metrics for the high-dimensional simulation. Error bars depict $\pm1$ standard deviation. ``Thresh.'' refers to the high-magnitude $\Phi$ condition where thresholdability is violated, and ``UCC'' refers to the condition where the unique child condition is violated.}
\label{fig:icml2}
\end{figure*}

The results are displayed in Figure~\ref{fig:icml2}.
We first note that in the thresholdability violation condition, thresholdability was indeed violated in at least 99\% of the data sets for each of the $n = 250$, 500, and 1000 configurations.
However, we can see the $F_1$ score become more accurate with $n$ despite these challenging conditions and the proportional growth in $p$.
The estimated number of latent variables $(\hat{d})$ was also fairly accurate on average across all conditions, confirming the CT algorithm is capable of determining the number of latent factors automatically even in such challenging high-dimensional settings.
Unsurprisingly, the computational time increased with $p$, but remained reasonable even at $p = 1500$.

\section{Real Data Application} \label{sec:real_data}

We examined a widely used factor analysis data set comprised of intelligence test scores of $n = 301$ middle school students \citep{Holzinger1939}, and compared the performance of the CT algorithm with EFA and the penalized EFA methods.
The data consist of 9 variables designed to measure 3 factors of intelligence.
These were a spatial factor, a verbal factor, and a speed factor.
The hypothesized structure of this design was an independent cluster structure between these three factors.
Again, for a fair comparison, we input the same set of $d$ values produced in the CT algorithm to each of the EFA methods as we did in the simulation studies.

We display the results in Table~\ref{realDataTable}.
We first checked the HD between the solution path of a method and the hypothesized model structure.
The minimum HD over the solution path was zero only for the CT algorithm, indicating that the hypothesized model was perfectly recovered within its solution path only, and not any other method.
Moreover, the CT algorithm identified the hypothesized structure with a much smaller set of candidate models.
We selected a structure via BIC for each method and used 10-fold CV to calculate a test-data log-likelihood and evaluate the structure.
The results for the test data log-likelihood are similar across all methods except EFA, which was much worse.
Despite the comparable performance between the CT algorithm and the sparse EFA methods, the CT algorithm obtained these results with much improved computational efficiency, measured by the number of candidate models evaluated.

\begin{table}[!t]
\centering
\begin{tabular}{cccccc}
\hline
\noalign{\vskip 2px}
Method & HD(min.) & $\hat{d}$ & Test LL & Models\\
\hline
\noalign{\vskip 2px}
CT Algorithm & 0 & 4  & -3749.60 & 13   \\
EFA          & 6 & 2  & -3823.14 & 4    \\
EFA-LASSO    & 6 & 3  & -3751.82 & 120  \\
EFA-MCP      & 3 & 4  & -3751.37 & 1080 \\
\hline
\noalign{\vskip 2px}
\end{tabular}
\caption{Results of real data example. HD(min.) denotes the minimum HD to the hypothesized structure across all solutions, Test LL refers to the test-data log-likelihood, and Models denotes the number of models evaluated per method.}
\label{realDataTable}
\end{table}

As with most factor analytic designs, the hypothesized structure followed the unique child condition.
An illustration of the hypothesized structure and all the selected structures by the four competing methods are displayed in Appendix~\ref{supp:figures} (Figure~\ref{fig:holzinger}).
Both the EFA-LASSO and EFA-MCP methods selected structures that followed the unique child condition, despite the fact that these methods are not developed under this assumption.
These convergent results lend empirical support for the unique child condition holding for factor analysis structures within this data set, as designed.

\section{Concluding Remarks} \label{sec:conclusion}

Overall, the CT algorithm is a promising method for learning factor analysis structures.
In this article, we motivated the algorithm using thresholded correlation graphs, and established conditions for the clique mapping procedure, parameter uniqueness, and asymptotic consistency.
In addition, the CT algorithm yields a method of learning $d$, which the EFA counterparts lack.
In our simulation studies, the CT algorithm performed nearly perfectly in low-dimensional settings, and showed robust results in high-dimensional settings.
Further, the computational efficiency of the CT algorithm is unrivaled relative to the EFA-LASSO and EFA-MCP methods, as it checks substantially less models.

There are some limitations of the CT algorithm, mainly the assumptions of thresholdability and the unique child condition.
While we demonstrated that the CT algorithm can be robust to violations of these assumptions in practice, our statistical consistency results depends on these assumptions being true in the population.
Future work can focus on the relaxation of these assumptions.

We also note some computational limitations for the high-dimensional ($n < p$) regime for parameter estimation.
Both penalized and traditional MLE estimation procedures have fairly long computation routines.
Since the CT algorithm relies on external existing estimation method to provide parameter estimates, it is subsequently limited by the existing technology in this area.
Thus the estimation portion of our algorithm will also benefit from computational advances on this topic.

\section*{Acknowledgments}

This work was supported by NSF grant DMS-1952929.

\bibliography{references}

\begin{thebibliography}{48}
\providecommand{\natexlab}[1]{#1}
\providecommand{\url}[1]{\texttt{#1}}
\expandafter\ifx\csname urlstyle\endcsname\relax
  \providecommand{\doi}[1]{doi: #1}\else
  \providecommand{\doi}{doi: \begingroup \urlstyle{rm}\Url}\fi

\bibitem[Anderson \& Gerbing(1988)Anderson and Gerbing]{Anderson1988}
Anderson, J.~C. and Gerbing, D.~W.
\newblock {Structural equation modeling in practice: A review and recommended
  two-step approach}.
\newblock \emph{Psychological Bulletin}, 103\penalty0 (3):\penalty0 411--423,
  1988.

\bibitem[Anderson \& Rubin(1956)Anderson and Rubin]{Anderson1956}
Anderson, T.~W. and Rubin, H.
\newblock {Statistical inference in factor analysis}.
\newblock In \emph{Proceedings of the Third Berkeley Symposium on Mathematical
  Statistics and Probability}, pp.\  111--150, 1956.

\bibitem[Auerswald \& Moshagen(2019)Auerswald and Moshagen]{Auerswald2019}
Auerswald, M. and Moshagen, M.
\newblock {How to determine the number of factors to retain in exploratory
  factor analysis: A comparison of extraction methods under realistic
  conditions}.
\newblock \emph{Psychological Methods}, 24\penalty0 (4):\penalty0 468--491,
  2019.

\bibitem[Browne(1968)]{Browne1968}
Browne, M.~W.
\newblock {A comparison of factor analytic techniques}.
\newblock \emph{Psychometrika}, 33\penalty0 (3):\penalty0 1968, 1968.

\bibitem[Browne(2001)]{Browne2001}
Browne, M.~W.
\newblock {An overview of analytic rotation in exploratory factor analysis}.
\newblock \emph{Multivariate Behavioral Research}, 36\penalty0 (1):\penalty0
  111--150, 2001.

\bibitem[Caner \& Han(2014)Caner and Han]{Caner2014}
Caner, M. and Han, X.
\newblock {Selecting the correct number of factors in approximate factor
  models: The large panel case with group bridge estimators}.
\newblock \emph{Journal of Business {\&} Economic Statistics}, 32\penalty0
  (3):\penalty0 359--374, 2014.

\bibitem[Cattell(1966)]{Cattell1966}
Cattell, R.~B.
\newblock {The Scree Test for the number of factors}.
\newblock \emph{Multivariate Behavioral Research}, 1\penalty0 (2):\penalty0
  245--276, 1966.

\bibitem[Crawford(1975)]{Crawford_1975}
Crawford, C.
\newblock A comparison of the direct oblimin and primary parsimony methods of
  oblique rotation.
\newblock \emph{British Journal of Mathematical and Statistical Psychology},
  28\penalty0 (2):\penalty0 201–--213, Nov 1975.
\newblock ISSN 00071102.
\newblock \doi{10.1111/j.2044-8317.1975.tb00563.x}.

\bibitem[Eppstein et~al.(2010)Eppstein, L{\"{o}}ffler, and
  Strash]{Eppstein2010}
Eppstein, D., L{\"{o}}ffler, M., and Strash, D.
\newblock {Listing all maximal cliques in sparse graphs in near-optimal time}.
\newblock In \emph{International Symposium on Algorithms and Computation}, pp.\
   403--414, 2010.

\bibitem[Ferguson(1996)]{Ferguson1996}
Ferguson, T.~S.
\newblock \emph{{A Course in Large Sample Theory}}.
\newblock Chapman {\&} Hall, 1996.

\bibitem[Ford et~al.(1986)Ford, MacCallum, and Tait]{Ford1986}
Ford, J.~K., MacCallum, R.~C., and Tait, M.
\newblock {The application of exploratory factor analysis in applied
  psychology: A critical review and analysis}.
\newblock \emph{Personnel Psychology}, 39:\penalty0 291--314, 1986.

\bibitem[Glorfeld(1995)]{Glorfeld1995}
Glorfeld, L.~W.
\newblock {An improvement on Horn's parallel analysis methodology for selecting
  the correct number of factors to retain}.
\newblock \emph{Educational and Psychological Measurement}, 55\penalty0
  (3):\penalty0 377--393, 1995.

\bibitem[Guttman(1954)]{Guttman1954}
Guttman, L.
\newblock {Some necessary conditions for common-factor analysis}.
\newblock \emph{Psychometrika}, 19\penalty0 (2):\penalty0 149--161, 1954.

\bibitem[Harris \& Kaiser(1964)Harris and Kaiser]{Harris1964}
Harris, C.~W. and Kaiser, H.~F.
\newblock {Oblique factor analytic solutions by orthogonal transformations}.
\newblock \emph{Psychometrika}, 29\penalty0 (4):\penalty0 347--362, 1964.

\bibitem[Hattie(1985)]{Hattie1985}
Hattie, J.
\newblock {Methodology review: Assessing unidimensionality of tests and items}.
\newblock \emph{Applied Psychological Measurement}, 9\penalty0 (2):\penalty0
  139--164, 1985.

\bibitem[Hirose \& Yamamoto(2014{\natexlab{a}})Hirose and
  Yamamoto]{Hirose2014a}
Hirose, K. and Yamamoto, M.
\newblock {Sparse estimation via nonconcave penalized likelihood in factor
  analysis model}.
\newblock \emph{Statistics and Computing}, 25\penalty0 (5):\penalty0 863--875,
  2014{\natexlab{a}}.

\bibitem[Hirose \& Yamamoto(2014{\natexlab{b}})Hirose and
  Yamamoto]{Hirose2014b}
Hirose, K. and Yamamoto, M.
\newblock {Estimation of an oblique structure via penalized likelihood factor
  analysis}.
\newblock \emph{Computational Statistics and Data Analysis}, 79:\penalty0
  120--132, 2014{\natexlab{b}}.

\bibitem[Holzinger \& Swineford(1939)Holzinger and Swineford]{Holzinger1939}
Holzinger, K.~J. and Swineford, F.
\newblock {A study in factor analysis: The stability of a bi-factor solution}.
\newblock \emph{Supplementary Educational Monographs}, 1939.

\bibitem[Horn(1965)]{Horn1965}
Horn, J.~L.
\newblock {A rationale and test for the number of factors in factor analysis}.
\newblock \emph{Psychometrika}, 30\penalty0 (2):\penalty0 179--185, 1965.

\bibitem[Howard(2016)]{Howard2016}
Howard, M.~C.
\newblock {A review of exploratory factor analysis decisions and overview of
  current practices: What we are doing and how can we improve?}
\newblock \emph{International Journal of Human-Computer Interaction},
  32\penalty0 (1):\penalty0 51--62, 2016.

\bibitem[Jennrich(2001)]{Jennrich2001}
Jennrich, R.~I.
\newblock {A simple general procedure for orthogonal rotation}.
\newblock \emph{Psychometrika}, 66\penalty0 (2):\penalty0 289--306, 2001.

\bibitem[Jennrich(2006)]{Jennrich2006}
Jennrich, R.~I.
\newblock {Rotation to simple loadings using component loss functions: The
  oblique Case}.
\newblock \emph{Psychometrika}, 71\penalty0 (1):\penalty0 173--191, 2006.

\bibitem[Jennrich \& Robinson(1969)Jennrich and Robinson]{Jennrich1969}
Jennrich, R.~I. and Robinson, S.~M.
\newblock {A Newton-Raphson algorithm for maximum likelihood factor analysis}.
\newblock \emph{Psychometrika}, 34\penalty0 (1):\penalty0 111--123, 1969.

\bibitem[J{\"{o}}reskog(1967)]{Joreskog1967}
J{\"{o}}reskog, K.~G.
\newblock {Some contributions to maximum likelihood factor analysis}.
\newblock \emph{Psychometrika}, 32\penalty0 (4):\penalty0 443--482, 1967.

\bibitem[Kaiser(1960)]{Kaiser1960}
Kaiser, H.~F.
\newblock {The application of electronic computers to factor analysis}.
\newblock \emph{Educational and Psychological Measurement}, 20\penalty0
  (1):\penalty0 141--151, 1960.

\bibitem[Kalisch \& B{\"{u}}hlmann(2007)Kalisch and
  B{\"{u}}hlmann]{Kalisch2007}
Kalisch, M. and B{\"{u}}hlmann, P.
\newblock {Estimating high-dimensional directed acyclic graphs with the
  PC-algorithm}.
\newblock \emph{Journal of Machine Learning Research}, 8:\penalty0 613--636,
  2007.

\bibitem[Kummerfeld \& Ramsey(2016)Kummerfeld and Ramsey]{Kummerfeld2016}
Kummerfeld, E. and Ramsey, J.
\newblock Causal clustering for 1-factor measurement models.
\newblock In \emph{Proceedings of the 22nd ACM SIGKDD International Conference
  on Knowledge Discovery and Data Mining}, pp.\  1655–1664, San Francisco
  California USA, Aug 2016. ACM.
\newblock ISBN 978-1-4503-4232-2.
\newblock \doi{10.1145/2939672.2939838}.
\newblock URL \url{https://dl.acm.org/doi/10.1145/2939672.2939838}.

\bibitem[Markham \& Grosse-Wentrup(2020)Markham and
  Grosse-Wentrup]{Markham2020}
Markham, A. and Grosse-Wentrup, M.
\newblock Measurement dependence inducing latent causal models.
\newblock In \emph{Proceedings of the 36th Conference on Uncertainty in
  Artificial Intelligence}, volume 124, 2020.

\bibitem[Mart{\'{i}}nez et~al.(1998)Mart{\'{i}}nez, Marshall, and
  Sechrest]{Martinez1998}
Mart{\'{i}}nez, M.~E., Marshall, J.~R., and Sechrest, L.
\newblock {Invited commentary: Factor analysis and the search for objectivity}.
\newblock \emph{American Journal of Epidemiology}, 148\penalty0 (1):\penalty0
  17--19, 1998.

\bibitem[Peeters(2012)]{Peeters2012}
Peeters, C. F.~W.
\newblock {Rotational uniqueness conditions under oblique factor correlation
  metric}.
\newblock \emph{Psychometrika}, 77\penalty0 (2):\penalty0 288--292, 2012.

\bibitem[Preacher et~al.(2013)Preacher, Zhang, Kim, and Mels]{Preacher2013}
Preacher, K.~J., Zhang, G., Kim, C., and Mels, G.
\newblock {Choosing the optimal number of factors in exploratory factor
  analysis: A model selection perspective}.
\newblock \emph{Multivariate Behavioral Research}, 48\penalty0 (1):\penalty0
  28--56, 2013.

\bibitem[{R Core Team}(2020)]{R2020}
{R Core Team}.
\newblock \emph{R: A Language and Environment for Statistical Computing}.
\newblock R Foundation for Statistical Computing, Vienna, Austria, 2020.
\newblock URL \url{https://www.R-project.org/}.

\bibitem[Ra{\^{i}}che et~al.(2013)Ra{\^{i}}che, Walls, Magis, Riopel, and
  Blais]{Raiche2013}
Ra{\^{i}}che, G., Walls, T.~A., Magis, D., Riopel, M., and Blais, J.-G.
\newblock {Non-graphical solutions for Cattell's Scree Test}.
\newblock \emph{Methodology}, 9\penalty0 (1):\penalty0 23--29, 2013.

\bibitem[Reise et~al.(2000)Reise, Waller, and Comrey]{Reise2000}
Reise, S.~P., Waller, N.~G., and Comrey, A.~L.
\newblock {Factor analysis and scale revision}.
\newblock \emph{Psychological Assessment}, 12\penalty0 (3):\penalty0 287--297,
  2000.

\bibitem[Revelle(2019)]{Revelle2019}
Revelle, W.
\newblock \emph{{psych}: Procedures for Psychological, Psychometric, and
  Personality Research}.
\newblock Northwestern University, Evanston, Illinois, 2019.
\newblock URL \url{https://CRAN.R-project.org/package=psych}.

\bibitem[Rosseel(2012)]{Rosseel2012}
Rosseel, Y.
\newblock {lavaan}: An {R} package for structural equation modeling.
\newblock \emph{Journal of Statistical Software}, 48\penalty0 (2):\penalty0
  1--36, 2012.
\newblock URL \url{http://www.jstatsoft.org/v48/i02/}.

\bibitem[Rubin \& Thayer(1982)Rubin and Thayer]{Rubin1982}
Rubin, D.~B. and Thayer, D.~T.
\newblock {EM Algorithms for ML Factor Analysis}.
\newblock \emph{Psychometrika}, 47\penalty0 (1):\penalty0 69--76, 1982.

\bibitem[Scharf \& Nestler(2019)Scharf and Nestler]{Scharf2019}
Scharf, F. and Nestler, S.
\newblock {Should regularization replace simple structure rotation in
  exploratory factor analysis?}
\newblock \emph{Structural Equation Modeling}, 26\penalty0 (4):\penalty0
  576--590, 2019.

\bibitem[Scheines et~al.(1998)Scheines, Spirtes, Glymour, Meek, and
  Richardson]{Scheines1998}
Scheines, R., Spirtes, P., Glymour, C., Meek, C., and Richardson, T.
\newblock {The TETRAD project: Constraint based aids to causal model
  specification}.
\newblock \emph{Multivariate Behavioral Research}, 33\penalty0 (1):\penalty0
  65--117, 1998.

\bibitem[Schreiber et~al.(2006)Schreiber, Nora, Stage, Barlow, and
  King]{Schreiber2006}
Schreiber, J.~B., Nora, A., Stage, F.~K., Barlow, E.~A., and King, J.
\newblock {Reporting structural equation modeling and confirmatory factor
  analysis results: A review}.
\newblock \emph{Journal of Educational Research}, 99\penalty0 (6):\penalty0
  323--338, 2006.

\bibitem[Shimizu et~al.(2009)Shimizu, Hoyer, and Hyvärinen]{Shimizu2009}
Shimizu, S., Hoyer, P.~O., and Hyvärinen, A.
\newblock Estimation of linear non-gaussian acyclic models for latent factors.
\newblock \emph{Neurocomputing}, 72\penalty0 (7–9):\penalty0 2024--–2027,
  Mar 2009.
\newblock ISSN 09252312.
\newblock \doi{10.1016/j.neucom.2008.11.018}.

\bibitem[Silva et~al.(2006)Silva, Scheines, Glymour, and Spirtes]{Silva2006}
Silva, R., Scheines, R., Glymour, C., and Spirtes, P.
\newblock {Learning the structure of linear latent variable models}.
\newblock \emph{Journal of Machine Learning Research}, 7:\penalty0 191--246,
  2006.

\bibitem[Tibshirani(1996)]{Tibshirani1996}
Tibshirani, R.
\newblock {Regression shrinkage and selection via the Lasso}.
\newblock \emph{Journal of the Royal Statistical Society: Series B},
  58\penalty0 (1):\penalty0 267--288, 1996.

\bibitem[Velicer \& Jackson(1990)Velicer and Jackson]{Velicer1990}
Velicer, W.~F. and Jackson, D.~N.
\newblock {Component analysis versus common factor analysis: Some issues in
  selecting an appropriate procedure}.
\newblock \emph{Multivariate Behavioral Research}, 25\penalty0 (1):\penalty0
  1--28, 1990.

\bibitem[Werts et~al.(1973)Werts, J{\"{o}}reskog, and Linn]{Werts1973}
Werts, C.~E., J{\"{o}}reskog, K.~G., and Linn, R.~L.
\newblock {Identification and estimation in path analysis with unmeasured
  variables}.
\newblock \emph{American Journal of Sociology}, 78\penalty0 (6):\penalty0
  1469--1484, 1973.

\bibitem[Whitely(1983)]{Whitely1983}
Whitely, S.~E.
\newblock {Construct validity: Construct representation versus nomothetic
  span}.
\newblock \emph{Psychological Bulletin}, 93\penalty0 (1):\penalty0 179--197,
  1983.

\bibitem[Zhang(2010)]{Zhang2010}
Zhang, C.~H.
\newblock {Nearly unbiased variable selection under minimax concave penalty}.
\newblock \emph{Annals of Statistics}, 38\penalty0 (2):\penalty0 894--942,
  2010.

\bibitem[Zwick \& Velicer(1986)Zwick and Velicer]{Zwick1986}
Zwick, W.~R. and Velicer, W.~F.
\newblock {Comparison of five rules for determining the number of components to
  retain}.
\newblock \emph{Psychological Bulletin}, 99:\penalty0 432--442, 1986.

\end{thebibliography}
\bibliographystyle{icml2023}

\newpage
\appendix
\onecolumn
\section{Proofs and Additional Results} \label{app:proofs}

\subsection{Proof of Lemma~\ref{lem:indep_maximal_clique}} \label{app:indep_maximal_clique}

\begin{proof}[\unskip\nopunct]
First, we prove that $\te{ne}(X_i)$ must be a maximal clique by contradiction.
Suppose $\te{ne}(X_i)$ is a clique, but not maximal.
Then $\te{ne}(X_i)$ can be extended by another node $X_j \notin \te{ne}(X_i)$, such that the union $X_j \cup \te{ne}(X_i)$ is a clique.
This implies that there is an edge between $X_i$ and $X_j$ and thus $X_j\in\te{ne}(X_i)$.
This leads to a contradiction, and therefore, $\te{ne}(X_i)$ must be maximal.
Second, we prove that $X_i$ is not a part of any other maximal clique, once again by contradiction.
Suppose that $X_i \in A$, where $A$ is a maximal clique and $A \neq \te{ne}(X_i)$.
By the definition of $\te{ne}(X_i)$, we must have $A \subset \text{ne}(X_i)$, i.e., a proper subset of $\text{ne}(X_i)$, which contradicts the hypothesis that $A$ is maximal.
Therefore, $X_i$ is not a part of any other maximal clique, making $\te{ne}(X_i)$ an independent maximal clique.
This completes the proof of the first statement.

Now we prove the second statement.
Let $X_i$ be any unique member of an independent maximal clique $C$. 
Suppose $\te{ne}(X_i)$ is not a subset of $C$, which means there is a vertex $X_j\notin C$ but is a neighbor of $X_i$.
Then $\{X_i,X_j\}$ either is a maximal clique or can be grown to a maximal clique $C'\ne C$.
This contradicts the fact that $X_i$ is a unique member of $C$.
Therefore, $\te{ne}(X_i)$ must be a subset of $C$ and thus is a clique. By the first statement of this lemma, $\te{ne}(X_i)$ is also an independent maximal clique and thus we must have $\te{ne}(X_i)=C$.
\end{proof}

\subsection{Necessary and Sufficient Condition for Thresholdability}

\begin{lemma} \label{generalSufficient}
Recall the definitions of $E_0$ and $E_0^c$ in Equations~\eqref{e_definition_E} and~\eqref{e_definition_Ec}, respectively.
A set of parameters $\theta$ is thresholdable if and only if:
\begin{equation} \label{gsIneq}
\underset{(k, l) \in E_0^c}{\max} \lvert \wt{\Lambda}_{kE} \Phi_{EF} \wt{\Lambda}_{lF}^T \rvert < \underset{(i, j)\in E_0}{\min} \lvert \wt{\Lambda}_{iA} \Phi_{AB} \wt{\Lambda}_{jB}^T + \wt{\Lambda}_{iC} \Phi_{CB} \wt{\Lambda}_{jB}^T + \wt{\Lambda}_{iA} \Phi_{AC} \wt{\Lambda}_{jC}^T + \wt{\Lambda}_{iC} \Phi_{CC} \wt{\Lambda}_{jC}^T \rvert,
\end{equation}
where $A = A(i, j) = \Pi_i - \Pi_j$, $B = B(i, j) = \Pi_j - \Pi_i$, $C = C(i, j) = \Pi_i \cap \Pi_j$, $E = \Pi_k$, and $F = \Pi_l$.
\end{lemma}

\begin{proof}
First it will be convenient to partition the parent variables of any pair $(X_i, X_j)$ as $\Pi_i \cup \Pi_j = \{L_A, L_B, L_C\}$, where:
\begin{equation}
\begin{aligned}
A &= \Pi_i - \Pi_j \\
B &= \Pi_j - \Pi_i \\
C &= \Pi_i \cap \Pi_j.
\end{aligned}
\end{equation}
Then we may re-cast Equation~\eqref{model} for any pair $(\widetilde{X}_i, \widetilde{X}_j)$ as follows:
\begin{equation}
\begin{bmatrix} \wt{X}_i^{\ps} \\[5pt] \wt{X}_j^{\ps} \end{bmatrix} =
\begin{bmatrix} \wt{\Lambda}_{iA} & \bm{0} & \wt{\Lambda}_{iC} \\[5pt] \bm{0} & \wt{\Lambda}_{jB} & \wt{\Lambda}_{jC} \end{bmatrix}
\begin{bmatrix} L_A^{\phantom{T}} \\[5pt] L_B^{\phantom{T}} \\[5pt] L_C^{\phantom{T}} \end{bmatrix} +
\begin{bmatrix} \wt{\epsilon}_i^{\ps} \\[5pt] \wt{\epsilon}_j^{\ps} \end{bmatrix}.
\end{equation}
We then obtain the correlation of between $X_i$ and $X_j$ from this form as follows:
\begin{equation}
\Var\left( \bbm \wt{X}_i^{\ps} \\[5pt] \wt{X}_j^{\ps} \ebm \right) =
\begin{bmatrix} \wt{\Lambda}_{iA} & \bm{0} & \wt{\Lambda}_{iC} \\[5pt] \bm{0} & \wt{\Lambda}_{jB} & \wt{\Lambda}_{jC} \end{bmatrix}
\begin{bmatrix} \Phi_{AA}^{\phantom{T}} & \Phi_{AB}^{\phantom{T}} & \Phi_{AC}^{\phantom{T}} \\[5pt] \Phi_{BA}^{\phantom{T}} & \Phi_{BB}^{\phantom{T}} & \Phi_{BC}^{\phantom{T}} \\[5pt] \Phi_{CA}^{\phantom{T}} & \Phi_{CB}^{\phantom{T}} & \Phi_{CC}^{\phantom{T}} \end{bmatrix}
\begin{bmatrix} \wt{\Lambda}_{iA}^T & \bm{0} \\[5pt] \bm{0} & \wt{\Lambda}_{jB}^T \\[5pt] \wt{\Lambda}_{iC}^T & \wt{\Lambda}_{jC}^T \end{bmatrix} +
\begin{bmatrix} \wt{\omega}_i^{\ps} & 0\\[5pt] 0 & \wt{\omega}_j^{\ps} \end{bmatrix},
\end{equation}
for which we multiply through and take the off-diagonal to be:
\begin{equation} \label{corEq}
\rho_{ij} = \wt{\Lambda}_{iA} \Phi_{AB} \wt{\Lambda}_{jB}^T + \wt{\Lambda}_{iC} \Phi_{CB} \wt{\Lambda}_{jB}^T + \wt{\Lambda}_{iA} \Phi_{AC} \wt{\Lambda}_{jC}^T + \wt{\Lambda}_{iC} \Phi_{CC} \wt{\Lambda}_{jC}^T.
\end{equation}
Writing $\rho_{ij}$ in this way yields a useful decomposition with respect to the structure of the factor analysis model.
Specifically, this can be thought of as the correlation between $X_i$ and $X_j$ due to their non-shared parents being correlated $(\Phi_{AB})$, their non-shared parents being correlated with their shared parents $(\Phi_{AC}, \Phi_{CB})$ and simply having shared parents $(\Phi_{CC})$.
Thus, if $X_i$ and $X_j$ have no shared parents, then the index set $C$ is empty.
This reduces Equation~\eqref{corEq} to:
\begin{equation} \label{corEqWc}
\rho_{ij} = \widetilde{\Lambda}_{iA} \Phi_{AB} \widetilde{\Lambda}_{jB}^T.
\end{equation}
The result of Lemma~\ref{generalSufficient} follows by characterizing the definition of thresholdability \eqref{e_thresholdability} directly in terms of $\theta$.
That is, if for all $(X_i, X_j)$ that share parents and for all $(X_k, X_l)$ that do not share parents, $\theta$ is thresholdable if and only if:
\begin{equation}
\underset{(k, l) \in E_0^c}{\max} \lvert \wt{\Lambda}_{kE} \Phi_{EF} \wt{\Lambda}_{lF}^T \rvert < \underset{(i, j) \in E_0}{\min} \lvert \wt{\Lambda}_{iA} \Phi_{AB} \wt{\Lambda}_{jB}^T + \wt{\Lambda}_{iC} \Phi_{CB} \wt{\Lambda}_{jB}^T + \wt{\Lambda}_{iA} \Phi_{AC} \wt{\Lambda}_{jC}^T + \wt{\Lambda}_{iC} \Phi_{CC} \wt{\Lambda}_{jC}^T \rvert.
\end{equation}
\end{proof}

\subsection{Proof of Corollary~\ref{cor:orthogonal_factors}}

\begin{proof}[\unskip\nopunct]
From Equation~\eqref{gsIneq}, we can see that if $\Phi = I_d$, then the $\Phi_{AB}$, $\Phi_{CB}$, $\Phi_{AC}$, and $\Phi_{EF}$ matrices are all zero matrices, and $\Phi_{CC}$ is an identity matrix.
Thus Equation~\eqref{gsIneq} reduces to
\begin{equation}
0 < \underset{(i, j) \in E_0}{\min} \lvert \wt{\Lambda}_{iC} \wt{\Lambda}_{jC}^T \rvert,
\end{equation}
which trivially holds.
\end{proof}

\subsection{Proof of Corollary~\ref{cor:independent_cluster}}

\begin{proof}[\unskip\nopunct]
The defining characteristic of the independent cluster structure is that $\Lambda$ has exactly one non-zero entry.
This implies that each observed variable has only one latent variable parent.
Thus, the relevant parent sets will reduce to $\Pi_i = \Pi_j = \{c\}$, $\Pi_k = \{e\}$, and $\Pi_l = \{f\}$.
That is, each pair of observed variables will either have one shared parent, or no shared parents, but not both.
Hence for each pair of variables that share parents, the $\Phi_{AB}$, $\Phi_{CB}$, and $\Phi_{AC}$ matrices will not exist and $\Phi_{CC} = 1$.
Corollary~\ref{cor:independent_cluster} follows by simplifying Equation~\eqref{gsIneq} with these reductions.
\end{proof}

\subsection{Proof of Lemma~\ref{lem:structural_identifiability}}

\begin{proof}[\unskip\nopunct]
Recall the definition of $E_0$, which we re-state for convenience:
$$E_0 \coloneqq \{ (i, j) : \Pi_i \cap \Pi_j \neq \emptyset \}.$$
Pick any $k\in[d]$.
By definition, every $X_j\in\te{ch}(L_k)$ shares a common parent $L_k$ and thus $\te{ch}(L_k)$ forms a clique in $\mathcal{G}=\mathcal{G}(X, E_0)$.
Let $U_k$ be the set of unique children of $L_k$.
Under the unique child condition, $U_k\ne \emptyset$, so we can pick an $X_i\in U_k$. Then $X_i$ does not have an edge connected to any node other than $\te{ch}(L_k)$ by the definition of $E_0$.
This implies every clique that includes $X_i$ must be a subset of $\te{ch}(L_k)$.
Thus, $\te{ch}(L_k)$ is the only maximal clique that includes $X_i$, making it an independent maximal clique.
The above argument shows that each $\te{ch}(L_k), k\in[d]$ is an independent maximal clique.
Since $\cup_k \te{ch}(L_k)=X$, any other maximal clique, if it exists, cannot be independent, and thus, $\{\te{ch}(L_k):k\in[d]\}$ is the set of independent maximal cliques in $\mathcal{G}$.
\end{proof}

\subsection{Formal Definition of Rotational Uniqueness} \label{sec:def_rot_unique}

\begin{definition}[Rotational Uniqueness]
For a set of parameters $\theta=\{\Lambda, \Phi, \Omega\}$, denote a rotated set of parameters as $\theta_M = \{\Lambda M, M^{-1} \Phi M^{-T}, \Omega\}$, where  $M$ is an invertible $d\times d$ matrix.
Let us define a set of \textit{constraint preserving rotations} as 
\begin{equation}
\mathcal{M}_{CP} = \mathcal{M}_{CP}(\theta) \coloneqq \{{M} : \Sigma(\theta_M)=\Sigma(\theta),
\mathcal{A}(\Lambda M) \subseteq \mathcal{A}(\Lambda), \text{diag}(M^{-1} \Phi M^{-T}) = I_d\}.
\end{equation}
Then:
\begin{enumerate}
  \item If $\mathcal{M}_{CP} = \{I_d\}$, then $\theta$ is said to be \textit{globally rotationally unique}.
  \item If $\mathcal{M}_{CP}$ is a set of signature matrices, then $\theta$ is said to be \textit{locally rotationally unique}, where signature matrices are diagonal matrices whose diagonal elements are $\pm 1$.
\end{enumerate}
\end{definition}
Note that the condition $\mathcal{A}(\Lambda M) \subseteq \mathcal{A}(\Lambda)$ ensures that the zero constraints implied by $\mathcal{A}(\Lambda)$ are persevered.
Two local rotational uniqueness properties relevant to the CT algorithm are described in Corollaries~\ref{cor:ucc_rotation} and~\ref{cor:ct_local_rotation_unique}.

\subsection{Proof of Corollary~\ref{cor:ucc_rotation}} \label{sec:ucc_rotation_proof}

\begin{proof}[\unskip\nopunct]

Define an index set for the rows of $\Lambda\in\mathbb{R}^{p\times d}$ which have zeroes in the $j$th column as
$$Z_j \coloneqq \{i: \lambda_{ij} = 0\}\subseteq [p],$$
and define
$$\Lambda^{[j]} \coloneqq \Lambda_{Z_j,-j},$$
which is a submatrix of size $|Z_j|\times (d-1)$.
Adapted from \cite{Peeters2012}, two sufficient conditions for $\Lambda$ that yield local rotational uniqueness for our model are:
\begin{description}
  \item[Condition 1:] $\Lambda$ has at least $d - 1$ zeroes in each column.
  \item[Condition 2:] $\text{rank}(\Lambda^{[j]}) = d - 1$ for all $j \in [d]$.
\end{description}
An example of $\Lambda^{[j]}$ is as follows:
\begin{equation}
\Lambda = \begin{bmatrix}
\lambda_{11} & 0 & 0 \\
\lambda_{21} & \lambda_{22} & 0 \\
\lambda_{31} & 0 & 0 \\
0 & \lambda_{42} & 0 \\
0 & \lambda_{52} & \lambda_{53} \\
0 & \lambda_{62} & 0 \\
0 & 0 & \lambda_{73} \\
0 & 0 & \lambda_{83} \\
\lambda_{91} & 0 & \lambda_{93} \\
\end{bmatrix}, \quad
\Lambda^{[1]} = \begin{bmatrix}
\lambda_{42} & 0 \\
\lambda_{52} & \lambda_{53} \\
\lambda_{62} & 0 \\
0 & \lambda_{73} \\
0 & \lambda_{83} \\
\end{bmatrix}, \quad
\Lambda^{[2]} = \begin{bmatrix}
\lambda_{11} & 0 \\
\lambda_{31} & 0 \\
0  & \lambda_{73} \\
0  & \lambda_{83} \\
\lambda_{91} & \lambda_{93} \\
\end{bmatrix}, \quad
\Lambda^{[3]} = \begin{bmatrix}
\lambda_{11} & 0 \\
\lambda_{21} & \lambda_{22} \\
\lambda_{31} & 0 \\
0 & \lambda_{42} \\
0 & \lambda_{62} \\
\end{bmatrix}.
\end{equation}

These conditions can be seen to be satisfied by the unique child condition as follows.
Let $U_j$ be the set of unique children for $L_j$ as defined in Equation~\eqref{eq:uc_def}.
For all $j, k \in [d]$, and $i \in [p]$ we can re-cast $U_j$ as:
\begin{equation}
U_j = \{ i: \lambda_{ij} \neq 0, \lambda_{ik} = 0, k \neq j\},
\end{equation}
and let the index of non-unique variables be:
\begin{equation}
\overline{U} = \{ i: i \notin \cup_{j = 1}^d U_j\}.
\end{equation}
Let us permute the rows of $\Lambda$ according to an order that satisfies $(U_1, \dots, U_d, \overline{U})$.
Denoting a permutation matrix that yields such a row ordering as $P$, we have:
\begin{equation} \label{eq:block_diagonal}
P\Lambda = \begin{bmatrix}
\Lambda_{U_1 1} & & \\
& \ddots & \\
& & \Lambda_{U_d d} \\
\Lambda_{\overline{U} 1} & \cdots & \Lambda_{\overline{U} d}
\end{bmatrix}.
\end{equation}
That is, we can permute the rows of $\Lambda$ such that its upper part is block-diagonal with $d$ blocks.
Then there must be at least $d - 1$ zeroes in each column, satisfying Condition 1.
It is easily seen that $P\Lambda$ also satisfies Condition 2, as any $(P\Lambda)^{[j]}$ will also have its upper part be block-diagonal, and thus full rank $(d - 1)$.
\end{proof}

\subsection{Proof of Corollary~\ref{cor:ct_local_rotation_unique}}

\begin{proof}
As described in Section~\ref{sec:ct_algorithm_specific}, Steps~6 through~9 of the CT algorithm construct the support $\hat{A}_k$ deterministically based on a set of independent maximal cliques $\mathcal{C}_k$ (from Step~5).
Since by Definition~\ref{def:independent_maximal_clique} independent maximal cliques always have a unique node, the sparsity pattern in $\hat{A}_k$ is guaranteed to follow the unique child condition (Condition~\ref{con:ucc}).
By Corollary~\ref{cor:ucc_rotation}, $\hat{\theta}_k$ will be locally rotationally unique due to this pattern.
\end{proof}

\subsection{Proof of Theorem~\ref{theo:consistency}} \label{sec:proof_consistency}

\begin{proof}[\unskip\nopunct]
To obtain our result, we will leverage existing estimation error bounds on the event $\lvert r_{ij} - \rho_{ij} \rvert \geq \epsilon$ for some $\epsilon > 0$.
To do this it will be convenient to re-cast our event of interest to $\hat{E}(\tau_0) \neq E_0$.
For clarity, let us first consider the event $\hat{E}(\tau_0) = E_0$, which by definition, holds if and only if:
\begin{equation}
\Big( \bigcap_{(i, j) \in E_0} \lvert r_{ij} \rvert > \tau_0 \Big) \cap \Big( \bigcap_{(i, j) \in E_0^c} \lvert r_{ij} \rvert < \tau_0 \Big).
\end{equation}
Then by De Morgan's laws, we can say $\hat{E}(\tau_0) \neq E$ if and only if:
\begin{equation}
\Big( \bigcup_{(i, j) \in E_0} \lvert r_{ij} \rvert \leq \tau_0 \Big) \cup \Big( \bigcup_{(i, j) \in E_0^c} \lvert r_{ij} \rvert \geq \tau_0 \Big),
\end{equation}
which is to say that $\hat{E}(\tau_0) \neq E_0$ holds if and only if any $r_{ij}$ is on the opposite side of $\tau_0$ as their population analog $\rho_{ij}$.
From here, the strategy is to derive bounds for $\Prob(\lvert r_{ij} \rvert \leq \tau_0)$ if $(i, j) \in E_0$, and $\Prob(\lvert r_{ij} \rvert \geq \tau_0)$ if $(i, j) \in E_0^c$, for all $(i, j)$.
To determine these bounds, we make use of a concentration inequality for $\Prob( \lvert r_{ij} - \rho_{ij} \rvert \geq \epsilon)$ from Lemma 1 of \citet{Kalisch2007}.
We re-state this as follows:

\begin{lemma} \label{l_KB07}
Assuming $X_i$ and $X_j$ are Gaussian random variables with correlation $\lvert \rho_{ij} \rvert \leq M < 1$.
Let $r_{ij}$ be the sample correlation calculated from an i.i.d. sample of size $n$. Then for any $0 < \epsilon \leq 2$, 
\begin{equation} \label{e_kblemma}
\Prob( \lvert r_{ij} - \rho_{ij} \rvert \geq \epsilon) \leq C_0(n - 2)\left( \dfrac{4 - \epsilon^2}{4 + \epsilon^2} \right)^{n - 4},
\end{equation}
where $0 < C_0 < \infty$ only depends on $M$.
\end{lemma}

For our purposes, we set $\epsilon = \gamma$ and select as $\tau_0$ the mid-point of $\min_{E_0}(\lvert \rho_{ij} \rvert)$ and $\max_{E_0^c}( \lvert \rho_{ij} \rvert)$, which will be the best choice to uniformly bound all $\Prob(\lvert r_{ij} \rvert \leq \tau_0)$ if $(i, j) \in E_0$ and $\Prob(\lvert r_{ij} \rvert \geq \tau_0)$ if $(i, j) \in E_0^c$.
The uniformity of the bound follows by seeing that $\gamma \leq \big\vert \lvert \rho_{ij} \rvert - \tau_0 \big\rvert$ for all $(i, j)$.
That is, there is no $\rho_{ij}$ that is closer to $\tau_0$ than the length of $\gamma$.

We begin with the scenario where $(i, j) \in E_0^c$.
Given the left-hand side of Equation~\eqref{e_kblemma} and setting $\epsilon = \gamma$, we have:
\begin{equation}
\begin{aligned}
\Prob( \lvert r_{ij} - \rho_{ij} \rvert \geq \gamma)
  &\geq \Prob( \lvert r_{ij} \rvert - \lvert \rho_{ij} \rvert \geq \gamma) \\
  &\geq \Prob( \lvert r_{ij} \rvert - \lvert \rho_{ij} \rvert \geq \tau_0 - \lvert \rho_{ij} \rvert) \\
  &= \Prob( \lvert r_{ij} \rvert \geq \tau_0).
\end{aligned}
\end{equation}
Hence, $\Prob( \lvert r_{ij} \rvert \geq \tau_0)$ is bounded from above by the right-hand side of Equation~\eqref{e_kblemma} if $(i, j) \in E_0^c$.
We can use the same strategy to conclude that, for $(i,j) \in E_0$,
\begin{equation}
\Prob( \lvert r_{ij} - \rho_{ij} \rvert \geq \gamma) \geq \Prob( \lvert r_{ij} \rvert \leq \tau_0).
\end{equation}

Since these two events have the same upper bound, let us combine them by defining:
\begin{equation}
 B_{ij} = B(r_{ij}, \tau_0) \coloneqq
  \begin{cases} 
      \lvert r_{ij} \rvert \leq \tau_0 & \text{if } (i, j) \in E_0 \\
      \lvert r_{ij} \rvert \geq \tau_0 & \text{if } (i, j) \in E_0^c
   \end{cases}.
\end{equation}
Noting that $\hat{E}(\tau_0) \neq E(\tau_0)$ holds if and only if $\bigcup_{(i, j)} B_{ij}$ holds, what remains is to find a bound of the latter event.
This can be done with the union bound:
\begin{equation} \label{eq:consistency_final_step}
\begin{aligned}
 \Prob(\hat{E}(\tau_0) \neq E(\tau_0))=\Prob \Big( \bigcup_{(i, j)} B_{ij} \Big) &\leq \sum_{(i, j)} \Prob \left(B_{ij} \right) \\
& \leq \dfrac{p(p - 1)}{2} \underset{(i, j)}{\max} \{ \Prob(B_{ij}) \}\\
 &\leq C p(p - 1) (n - 2)\left( \dfrac{4 - \gamma^2}{4 + \gamma^2} \right)^{n - 4},
\end{aligned}
\end{equation}
where $0 < C < \infty$ only depends on $M$.
This result follows by recognizing that all $\Prob(B_{ij})$ are uniformly bounded as in Lemma~\ref{l_KB07}.
Finally, this implies
\begin{equation}
\Prob(\hat{E}(\tau_0) = E_0) \geq 1 - C p(p - 1) (n - 2)\left( \dfrac{4 - \gamma^2}{4 + \gamma^2} \right)^{n - 4}
\end{equation}
and thus, \eqref{eq:support_bound} follows immediately under the unique child condition by Lemma~\ref{lem:structural_identifiability}.
\end{proof}

\subsection{Proof of Theorem~\ref{theo:hd_consistency}} \label{sec:proof_hds}

\begin{proof}[\unskip\nopunct]
To begin, we will first examine the growth of a lower bound of $\Prob(\hat{E}(\tau_0) = E_0)$ as a function of $n$.
Noting from Equation~\eqref{eq:consistency_bound}, an upper bound on the decaying term with $n$ can be derived as follows:
\begin{equation}
\begin{aligned}
\left( \dfrac{4 - \gamma^2}{4 + \gamma^2} \right)^{n - 4} &\leq \left( 1 - \dfrac{\gamma^2}{4} \right)^{n - 4} \\
&\leq \left( 1 - \dfrac{c_1}{4(n - 4)^b} \right)^{n - 4} \\
&= \left( 1 - \dfrac{c_1}{4(n - 4)^b} \right)^{(n - 4)^b(n - 4)^{1 - b}} \\
&= \left( \exp\left( -\dfrac{c_1}{4} \right) + o(1) \right)^{(n - 4)^{1 - b}} \\
&\leq \exp\left( -\dfrac{c_2(n - 4)^{1 - b}}{4} \right),
\end{aligned}
\end{equation}
where we used the limit $\lim_{x \to \infty} (1 + a/x)^x = \exp(a)$ and another constant $c_2 \in (0, c_1)$ such that the $o(1)$ remainder can be dropped.
From here, we can form a looser bound on Equation~\eqref{eq:consistency_bound} as
\begin{equation}
\begin{aligned}
\Prob(\hat{E}(\tau_0) = E_0) &\geq 1 - C p(p - 1) (n - 2)\left( \dfrac{4 - \gamma^2}{4 + \gamma^2} \right)^{n - 4} \\
&\geq 1 - C p_n(p_n - 1) (n - 2) \exp\left( -\dfrac{c_2(n - 4)^{1 - b}}{4} \right) \\
&= 1 - p(n)f(n),
\end{aligned}
\end{equation}
where $p(n) =  p_n(p_n - 1)$ and $f(n) = (n - 2) \exp( -c_2(n - 4)^{1 - b}/4)$.
Therefore, we have consistency if $\lim_{n \to \infty} p(n)f(b) = 0$ or if $p(n) = o(1/f(n))$.
Comparing the dominating terms of $p(n)$ and $1/f(n)$, consistency is achieved if
\begin{equation}
\begin{aligned}
p^2_n &= o\left( \exp\left[ \dfrac{c_2(n - 4)^{1 - b}}{4}-\log n \right] \right)\\
\text{or if } p_n &= o\left( \exp\left[c(n - 4)^{1 - b} \right] \right),
\end{aligned}
\end{equation}
by choosing a positive constant $c < c_2/8$.
\end{proof}

\newpage

\section{Supplementary Details for Simulation Studies and Real Data Application} \label{supp:simulation}

\subsection{Simulation settings}

The simulations were done in the \texttt{R} language \citep[4.0.2;][]{R2020}.
The \texttt{lavaan} package \citep{Rosseel2012} was used in the estimation phases of the CT algorithm (Step~10), and was used to estimate the baseline MLE solution.
For the cutoffs $\tau_k$, 40 equidistant points from 0 to 1 were input for the CT algorithm.
For EFA, the \texttt{psych} package \citep{Revelle2019} was used to obtain MLE solutions for unconstrained $\Lambda$.
We left the rotation option to the package default oblimin method \citep{Crawford_1975}, however we note that the rotation choice does not affect the results since we will only be examining the likelihood of $\Sigma(\hat{\theta})$.
And finally, the LASSO and MCP variants of EFA were estimated with the \texttt{fanc} package \citep{Hirose2014a, Hirose2014b}.
The tuning parameters were left at the package defaults of 30 values for a single tuning parameter in LASSO and 270 combinations of two tuning parameters in MCP.
For estimating the number of latent factors in these methods, the number of non-zero columns in $\hat{\Lambda}$ was taken as $\hat{d}$ as they would serve as the de facto number of latent variables \citep{Caner2014}.

To generate $\Phi$, we began by setting its diagonals to one.
Then for the off-diagonal elements, we drew a $d\times d$ matrix $A$ with entries from $\text{Uniform}(0, 1)$ and rescaled it such that $A^TA$ had off-diagonals in the range of $[0.6, 0.8]$, the range of $\lambda_{ij}$.
Then the off-diagonals of $\Phi$ were set to the off-diagonals of this rescaled $A^TA$, which ensured $\Phi$ would be positive definite.
Then depending on the condition for the magnitude of $\Phi$, the off-diagonals were scaled by 0.25 for the low $\Phi$ condition and by 0.75 for the high $\Phi$ condition.

For the data sets that violate the unique child condition, we began with an independent cluster structure (one non-zero entry per row), for which the unique child condition trivially holds for every latent variable.
We will call these latent variables the \textit{main parent} of these observed variables.
To isolate the effect of the unique child condition from that of thresholdability, we ensured thresholdability was always met in the population by setting $\Phi = I_d$ (Corollary~\ref{cor:orthogonal_factors}).
Then we chose 75\% of the latent variables at random to have no unique children.
If a latent variable was deemed to have no unique children, we generated an extra path between all the children of this latent variable to another random latent variable.
We will call these parents the \textit{extra parent}.

For each $X_j$, we drew an $R^2 \sim \text{Uniform}(0.36, 0.64)$ as the proportion of variance in $X_j$ explained by $L$.
The range of (0.36, 0.64) is analogous to the range of path coefficients we were using in previous simulations which was (0.6, 0.8).
If a given $X_j$ only had a main parent and no extra parent, then that $X_j$ had a single path coefficient of $\sqrt{R^2}$ from its main parent.
However, if a given $X_j$ also had an extra parent, then the $R^2$ was split using a 5:1 ratio between the main parent and the extra parent, and the path coefficients were calculated to reflect this accordingly.

\subsection{Evaluation Metrics} \label{sec:outcomes}

To compare the estimated and true supports ($\mathcal{A}(\hat{\Lambda})$ vs. $\mathcal{A}(\Lambda)$) we computed the minimum HD over all column permutations of $\hat{\Lambda}$.
That is, we define an HD as
\begin{equation} \label{eq:hd}
\text{HD} \coloneqq \min_{P} \left[ \lvert \mathcal{A}(\hat{\Lambda} P) \, \triangle \, \mathcal{A}(\Lambda) \rvert \right],
\end{equation}
where $\triangle$ is the symmetric difference or disjunctive union between two sets.
The permutation matrix $P$ reconciles the fact that the column order of $\hat{\Lambda}$ may not be the same as the column order of $\Lambda$, and that $\hat{d}$ may not be the same as $d$.
Put another way, HD is the smallest number of element additions and deletions needed to make the sets $\mathcal{A}(\Lambda)$ and $\mathcal{A}(\hat{\Lambda})$ identical, among all column permutations of $\hat{\Lambda}$.

In addition to HD, we also report the $F_1$ score, a normed measure of classification.
This allows for comparability between models with differing dimensions of $\Lambda$, that is differing $p$ and $d$.
Note that the $F_1$ score is simply the harmonic mean between precision and recall.
Once again using a permutation matrices to reconcile different orderings of $L$, we have
\begin{equation} \label{eq:f_1}
F_1(\hat{\Lambda}) \coloneqq \max_{P} \left[ \dfrac{2 \lvert \mathcal{A}(\hat{\Lambda}P) \cap \mathcal{A}(\Lambda) \rvert}{2 \lvert \mathcal{A}(\hat{\Lambda}P) \cap \mathcal{A}(\Lambda) \rvert + \lvert \mathcal{A}(\hat{\Lambda}P) \, \triangle \, \mathcal{A}(\Lambda) \rvert} \right] \in [0,1],
\end{equation}
and the higher the $F_1$ score, the more accurate the estimated support of $\hat{\Lambda}$.

\newpage

\section{Additional Figures} \label{supp:figures}

\begin{figure*}[!h]
\centering
\includegraphics[width=0.9\textwidth, keepaspectratio]{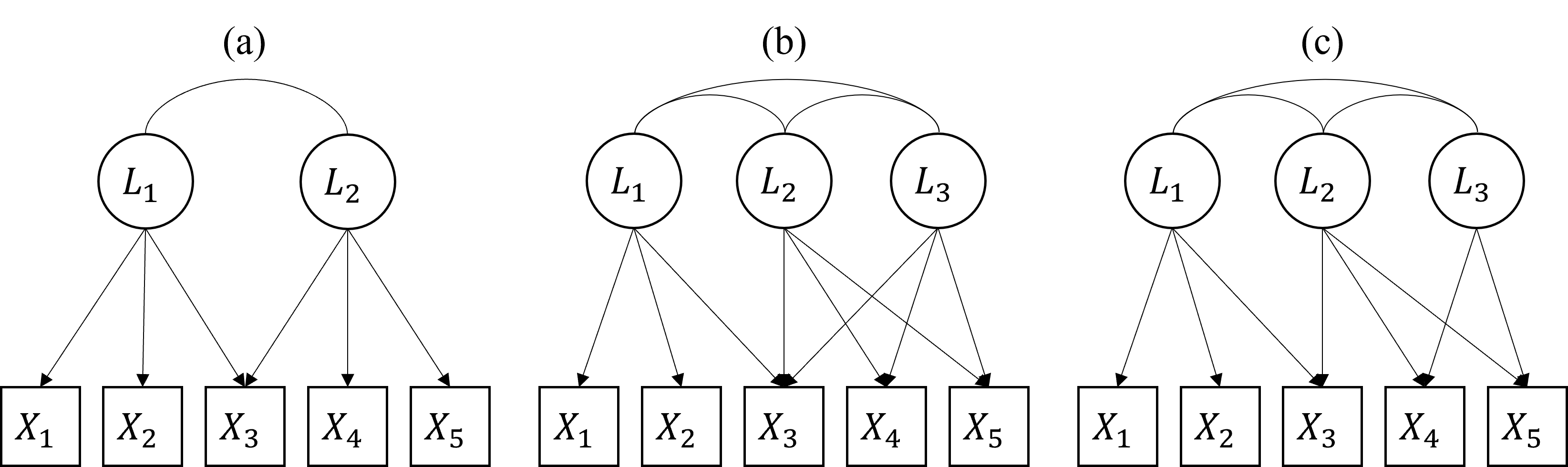}\\
\caption{Example structures that may be precluded from the unique child condition. Figure (a) is a structure that meets the unique child condition, Figure (b) shows a structure where the children sets of $L_2$ and $L_3$ are identical, and Figure (c) shows a structure where the children set of $L_3$ is a proper subset of $L_2$. All three structures lead to the same thresholded correlation graph, and will contain two independent maximal cliques.}
\label{fig:ucc_counterexamples}
\end{figure*}

\begin{figure*}[!h]
\centering
\includegraphics[width=0.9\textwidth, keepaspectratio]{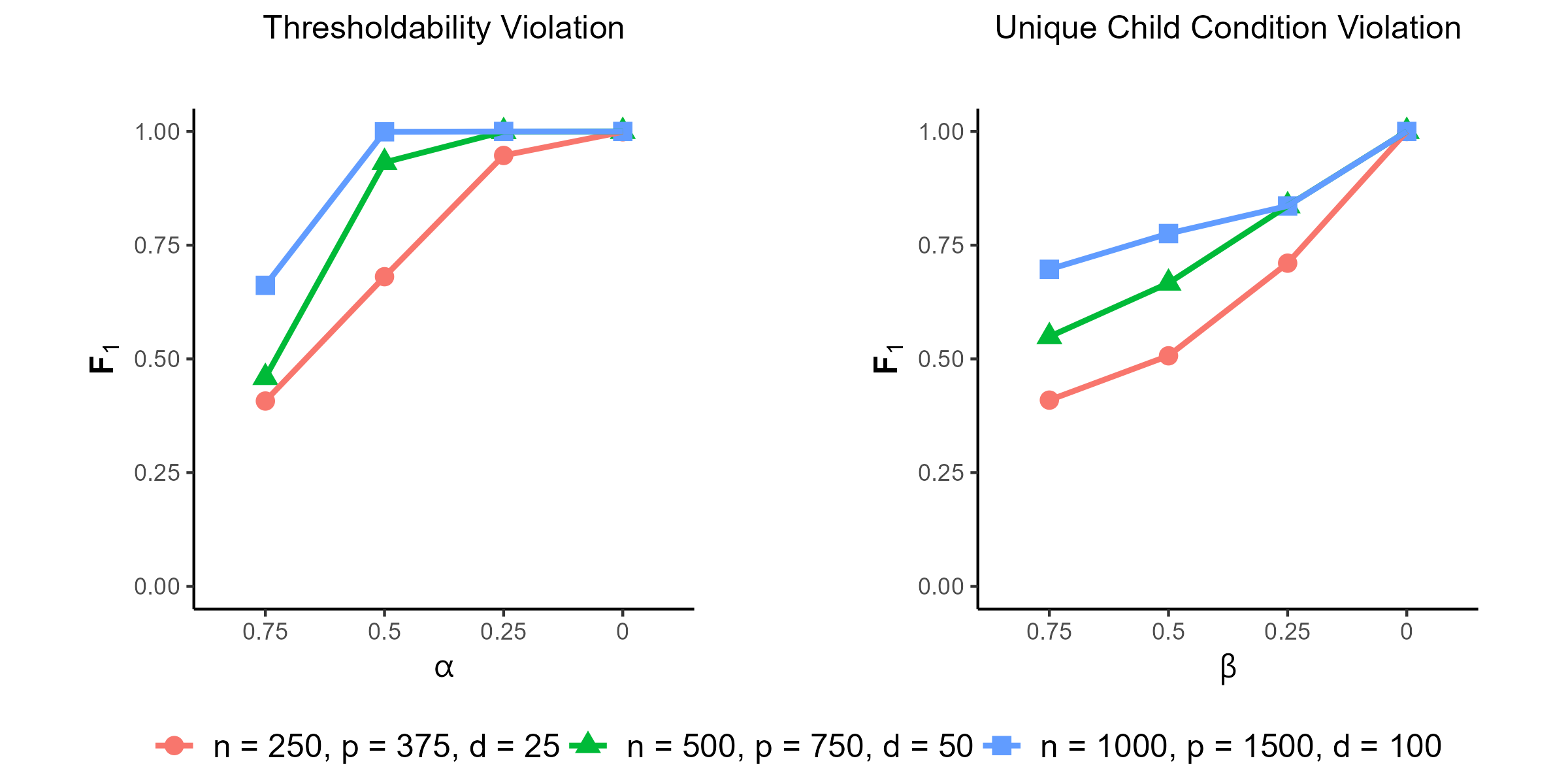}\\
\caption{Trends for $F_1$ score when each of the assumptions are violated to varying degrees in the high-dimensional setting. Thresholdability was varied via the scaling factor $\alpha \in [0, 0.75]$ on the off-diagonals in $\Phi$. The proportion of latent parents without unique children is represented by $\beta \in [0, 0.75]$.}
\label{fig:thresh_ucc_range}
\end{figure*}

\begin{figure*}[!h]
\centering
\includegraphics[width=0.9\textwidth, keepaspectratio]{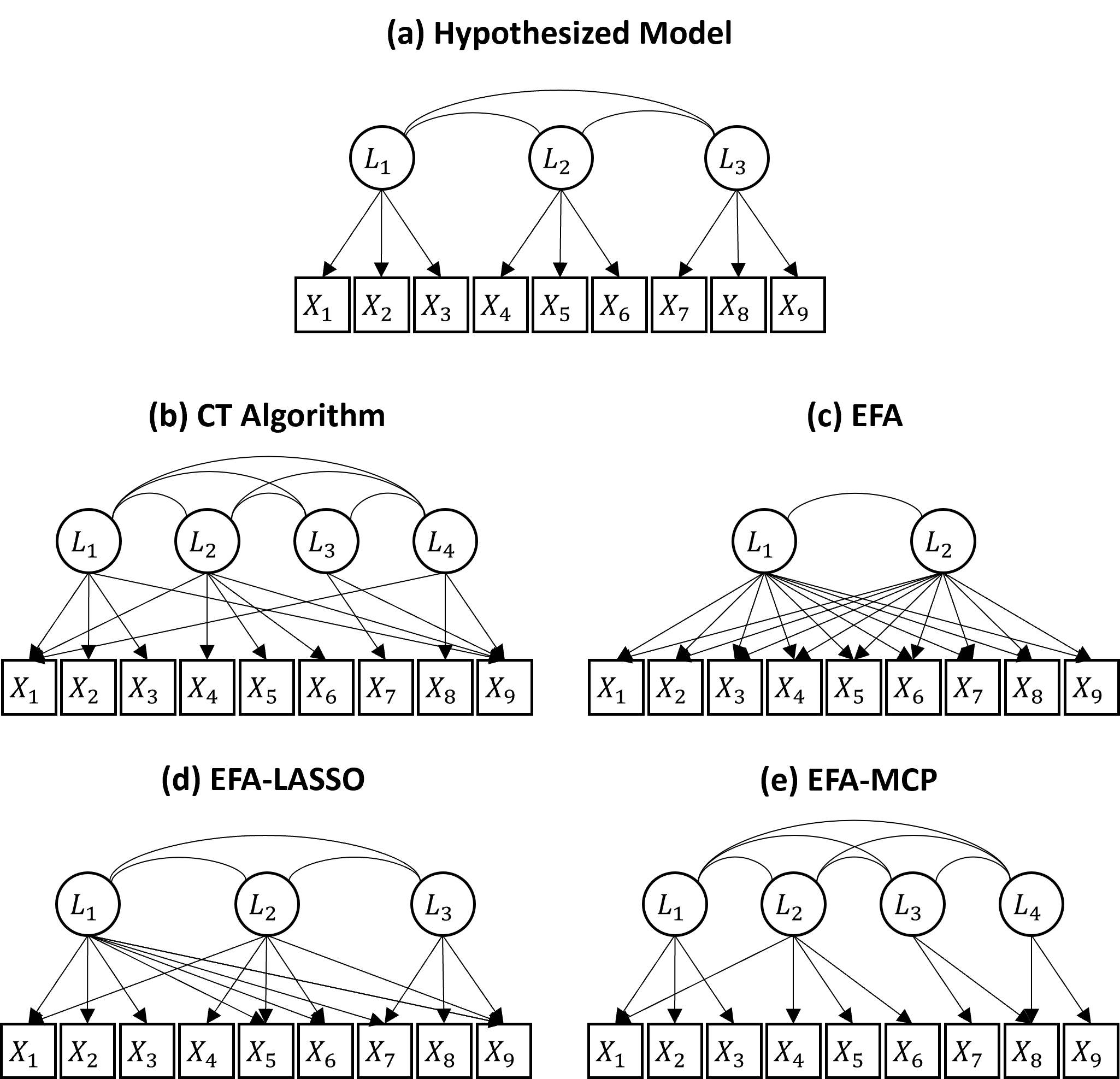}\\
\caption{The hypothesized and estimated model structures by each method in the real data example. Variables $X_1$, $X_2$, and $X_3$ were visual perception tasks, variables $X_4$, $X_5$, and $X_6$ were verbal/reading tasks, and variables $X_7$, $X_8$, and $X_9$ were speed tests.}
\label{fig:holzinger}
\end{figure*}


\end{document}